\documentclass[twocolumn,unpublished, a4paper]{quantumarticle}
\pdfoutput=1

\usepackage[utf8]{inputenc}
\usepackage[english]{babel}
\usepackage[T1]{fontenc}

\usepackage{tikz}
\usepackage{lipsum}
\usepackage{graphicx}%
\usepackage{multirow}%
\usepackage{amsmath,amssymb,amsfonts}%
\usepackage{amsthm}%
\usepackage{mathrsfs}%
\usepackage{xcolor}%
\usepackage{float}
\usepackage{listings}%
\usepackage{longtable}

\usepackage{orcidlink} 
\usepackage{algorithm}
 
\usepackage{algpseudocode}
\usepackage{hyperref}

\date{Publication Date}

\makeatletter
\newenvironment{breakablealgorithm}
  {%
   \begin{center}
   \refstepcounter{algorithm}%
   \hrule height.8pt depth0pt \kern2pt
   \renewcommand{\caption}[2][\relax]{%
     {\raggedright\textbf{\fname@algorithm~\thealgorithm} ##2\par}%
     \ifx\relax##1\relax
       \addcontentsline{loa}{algorithm}{\protect\numberline{\thealgorithm}##2}%
     \else
       \addcontentsline{loa}{algorithm}{\protect\numberline{\thealgorithm}##1}%
     \fi
     \kern2pt\hrule\kern2pt
   }%
  }
  {%
   \kern2pt\hrule\relax
   \end{center}
  }
\makeatother

\newtheorem{theorem}{Theorem}

\begin{document}

\title{Quantum Homomorphic Encryption: Towards Practical QOTP-Encrypted Computation on Gate-Based Hardware}

\author{Jon Hernández-Bueno}
\affiliation{TECNALIA, Basque Research and Technology Alliance (BRTA), Bizkaia Science and Technology Park, Building 700, E-48160 Derio, Bizkaia, Spain.}
\affiliation{Department of Communications Engineering, University of the Basque Country (UPV/EHU). Faculty of Engineering of Bilbao, Plaza Ingeniero Torres Quevedo, n.1, Bilbao, 48013, Spain.}
\orcid{0009-0009-6640-0983}\email{jon.hernandez@tecnalia.com}
\author{Oscar Lage}
\orcid{0000-0003-1168-1932}
\affiliation{TECNALIA, Basque Research and Technology Alliance (BRTA), Bizkaia Science and Technology Park, Building 700, E-48160 Derio, Bizkaia, Spain.}
\author{Marivi Higuero}
\affiliation{Department of Communications Engineering, University of the Basque Country (UPV/EHU). Faculty of Engineering of Bilbao, Plaza Ingeniero Torres Quevedo, n.1, Bilbao, 48013, Spain.}
\orcid{0000-0001-8451-556X}
\author{Jasone Astorga}
\affiliation{Department of Communications Engineering, University of the Basque Country (UPV/EHU). Faculty of Engineering of Bilbao, Plaza Ingeniero Torres Quevedo, n.1, Bilbao, 48013, Spain.}
\affiliation{EHU Quantum Center, University of the Basque Country (UPV/EHU). Faculty of Science and Technology, Barrio de Sarriena s/n, Leioa, 48940, Spain.}
\orcid{0000-0002-5532-004X}

\maketitle

\begin{abstract}

As quantum computing moves toward remote and cloud-based access, protecting quantum data during outsourced execution becomes important. Quantum one-time pad (QOTP) encryption provides information-theoretic hiding of quantum state, but evaluating gates on Pauli-encrypted data is nontrivial: Clifford operations allow classical Pauli-frame updates, whereas non-Clifford operations require key-dependent corrections.

We introduce QOTPH, a client-side key-tracked compilation framework for encrypted quantum computation. The client retains the QOTP keys and transforms each operation into a corrected gate or decomposition. The resulting circuit is executed non-adaptively by the quantum server, which neither receives the keys nor performs key updates.

We develop a unified Pauli frame compilation layer spanning Clifford operations, one- and two-qubit Pauli-generated rotations, explicit key-dependent parametrized controlled rotations, and recursively composed gates, and establish a global correctness result for circuits composed of the formally verified primitive rules. We also characterize the evaluator's view under a hidden logical circuit model, in which the QOTP frame, the uncompiled logical circuit, and the logical-to-physical instruction mapping remain client-side; the security of this model against various algorithms is also analyzed. The implementation is evaluated in Qiskit using noiseless simulation and IBM Quantum hardware, including in-circuit and local-decryption workflows.

This work bridges the gap between algebraic QOTP key tracking and practical encrypted quantum computation by providing a unified, hardware-validated compilation framework for current gate-based quantum processors.

\end{abstract}

\section{Introduction}
\label{sec:Intro}

Modern cryptography increasingly requires computing on data while keeping it encrypted, providing correct results from the processing of such data without revealing the underlying information, a requirement that also arises in quantum computation.

Homomorphic encryption~\cite{gong2024practical} currently represents one of the most promising research directions in post-quantum cryptography and computational privacy~\cite{first3finalized}. It is an emerging cryptographic paradigm that enables the evaluation of arbitrary functions on encrypted data, preserving the confidentiality of the information even during its processing. This class of cryptosystems is particularly relevant in scenarios involving delegated computation over untrusted environments, such as cloud computing services. Unlike traditional encryption schemes, homomorphic encryption allows a third party to process encrypted information without access to the secret keys, thereby offering both a theoretical and practical framework for secure computation without data disclosure.

The need for such cryptographic mechanisms has become increasingly critical in recent years due to the exponential growth of data processing in distributed environments and the rising concern over technological sovereignty, digital privacy, and the exposure of sensitive information, for example, in international contexts marked by geopolitical tensions~\cite{shoker2022digital}. As a result, research on homomorphic encryption schemes, such as BFV~\cite{fan2012somewhat}, BGV~\cite{brakerski2014leveled}, CKKS~\cite{cheon2017homomorphic}, and GSW~\cite{gentry2013homomorphic}, has gained significant momentum, both from a theoretical standpoint and in terms of practical implementations and algorithmic optimizations~\cite{InsiderHEMarket2024}.

In parallel with the development of advanced cryptography in the classical domain, quantum computing has undergone significant progress over the past decades, with the promise of revolutionizing multiple areas of computation and information security. This emerging computational paradigm has prompted a profound reassessment of the foundations of modern cryptography.

Two quantum algorithms clearly illustrate this impact. On the one hand, Shor's algorithm~\cite{shor1999polynomial} enables the polynomial-time factorization of integers and the computation of discrete logarithms on a quantum computer, directly undermining the security of classical schemes such as RSA, DSA, and ECC. On the other hand, Grover's algorithm~\cite{grover1996fast} offers a quadratically accelerated search over unstructured domains, reducing the effective security of symmetric encryption and hash functions by decreasing the brute-force complexity from $O(N)$ to $O(\sqrt{N})$.

Due to such breakthroughs, the necessity of implementing quantum-secure or post-quantum cryptographic schemes is becoming increasingly urgent, with the goal of preserving data security and privacy in the advent of quantum adversaries.

In this article, we address this challenge by developing and experimentally validating a QOTP-based key-tracked compilation framework for encrypted quantum computation. Our work is organized as follows: Section~\ref{sec:RelatedWork} reviews previous QHE constructions and positions QOTPH with respect to them. In Section~\ref{sec:Backg}, we present the fundamental concepts involved, briefly reviewing classical homomorphic encryption and the current state of quantum computing, including previous quantum encryption schemes. Section~\ref{sec:Approach} details our proposal, describing the development of the model. In Section~\ref{sec:Implementation}, we present the practical implementation of the model, providing illustrative pseudocode and a clear explanation of the experimental environment. The results obtained, together with their detailed analysis, are presented in Section~\ref{sec:Results}. Section~\ref{sec:security} analyzes the security properties and privacy boundaries of the direct QOTPH execution model, including potential QOTP key-inference channels when the evaluator has algorithmic or instance-specific side information. Finally, in Section~\ref{sec:Conclu}, we summarize the main conclusions, highlight the key contributions of this work, and discuss future implications and possible lines of research arising from this study.

\subsection{Contributions and scope}
\label{subsec:contributions}

The main contributions of this work are as follows:

First, we formulate QOTPH as a client-side key-tracked compilation framework. The client generates and retains the QOTP key frame, resolves all key-dependent gate corrections locally, optionally incorporates the initial $X$-mask into computational-basis input data, and submits a non-adaptive physical circuit to the quantum server.

Second, we provide a unified and formally verified gate-compilation
layer combining Clifford propagation, one- and two-qubit Pauli-generated rotations, explicit key-dependent rules for parametrized controlled rotations, and composite gates obtained recursively from the verified primitive set within a single Pauli-frame formalism. We establish a global correctness theorem showing that decryption with the final QOTP frame recovers the intended logical computation.

Third, we implement the compilation workflow in Qiskit and experimentally compare the plaintext and compiled encrypted circuits using noiseless simulation and gate-based IBM Quantum hardware. We additionally compare in-circuit decryption with local decryption of encrypted measurement outcomes.

Fourth, we characterize the evaluator model and the privacy boundary of the compilation process. In the delegated workflow, the QOTP frame, the logical circuit, and the correspondence between logical and physical instructions remain client-side, while the server receives only the flattened compiled circuit. Under this model, individual corrected operations are compatible with multiple logical-gate and key-frame configurations and therefore do not directly encode a unique QOTP key. Additional knowledge of the exact logical circuit, including relevant signed parameters and their correspondence with compiled instructions, can reduce this ambiguity and may reveal QOTP-frame information; this side-information setting is analyzed explicitly in Section~\ref{subsec:algorithm-aware-analysis}.

\section{Related Work} 
\label{sec:RelatedWork}

Within this broader context, previous research has demonstrated that quantum computing can also enhance cryptographic security and randomness. Regarding this second aspect, truly random keys have been successfully generated using certified quantum devices~\cite{hernandez2024true}, which overcome the limitations of classical pseudo-random entropy sources. Additionally, the feasibility of Quantum Fingerprinting has been proposed and validated in distributed systems such as blockchains, offering efficient data comparison methods and quantum-based mechanisms for verifying the integrity of blocks and transactions~\cite{hernandez2025enhancing}.

In this context, where quantum computing emerges both as a cryptographic adversary and an unexpected ally, data has become potentially vulnerable in the short and long term. As a consequence, a critical need arises: the development of a Quantum Homomorphic Encryption (QHE) framework~\cite{dulek2016quantum}. Despite substantial theoretical progress, practical hardware-oriented QFHE with strong end-to-end security guarantees remains challenging, particularly when compactness, non-Clifford evaluation, client requirements, and realistic implementation overhead are considered simultaneously, making this area one of the major open challenges in modern quantum cybersecurity.

Recent advances in quantum homomorphic encryption illustrate the dynamic intersection between quantum computation and privacy-preserving cryptography. QHE aspires to empower users with the ability to delegate quantum computations on encrypted data, potentially unlocking the full promise of quantum cloud computing without compromising confidentiality.

A significant advance in the experimental realization of QHE was achieved by Li et al.~\cite{li2024experimental}, who demonstrated a proof-of-concept implementation of quantum homomorphic operations on a compact photonic chip. By leveraging integrated photonic technology, their proof-of-principle experiment established that simple quantum circuits can indeed be evaluated on encrypted quantum data in a physical device. While this result marks a significant leap from theory to practice, it is worth noting that the demonstrated circuits remain limited in complexity, and key challenges concerning scalability and fault tolerance persist unresolved.

On the theoretical front, substantial progress has been made by Gupte and Vaikuntanathan~\cite{gupte2024construct}, who proposed a generic construction of compact quantum fully homomorphic encryption from classical fully homomorphic encryption with decryption in $\mathsf{NC}^{1}$, together with dual-mode trapdoor functions. Their framework is notable for enabling classical clients to outsource quantum computations, thereby broadening the scope of practical applicability. However, the construction relies on substantially stronger cryptographic machinery than the direct QOTP key-tracked compilation considered here and is primarily developed in a leveled QFHE setting, rather than as a near-term hardware-oriented implementation.

On the algorithmic side, efforts to bridge theory and practice have produced encouraging results. Fernández and Martin-Delgado~\cite{fernandez2024implementing} conducted simulations of Grover’s search algorithm executed homomorphically under a quantum encryption scheme using quantum circuits. Their findings suggest that, at least for small circuits, the overhead introduced by encryption is manageable, and further optimizations are possible for specific scenarios, such as the unique-solution case. Nevertheless, these studies are confined to idealized environments, leaving the challenges posed by noisy, large-scale quantum hardware unaddressed.

To further tackle practical resource bottlenecks, Ortega, Fernández, and Martin-Delgado~\cite{ortega2025implementing} developed a semiclassical approach for evaluating Szegedy quantum walks in a homomorphic setting. By dynamically generating one-time pad keys during execution, they significantly reduce the exponential key management burden associated with earlier protocols. Despite these advances, their approach relies on mid-circuit measurement, qubit reset, and classical-control capabilities and therefore follows a different execution model from the static non-adaptive circuits targeted by QOTPH.

Finally, the question of efficient ciphertext retrieval within QHE frameworks has been addressed by Cheng, Chen, and Wang~\cite{cheng2025secure}, who introduced a protocol involving a trusted third party (TTP) to enhance the performance of homomorphic Grover search. While their design reduces the ancillary resources required and demonstrates feasibility in simulation, it does so at the expense of introducing additional trust assumptions and communication overhead, which could pose limitations in adversarial or distributed contexts.

At the gate level, earlier work already established several ingredients used in QOTP-based encrypted evaluation. Liang showed that QOTP-based evaluation may depend on the secret key and extended such key-dependent evaluation to arbitrary unitary transformations~\cite{liang2013symmetric}. Ma and Li later considered arbitrary single-qubit gates through a different QFHE construction based on an extension of the Pauli one-time pad~\cite{ma2022quaternion}. Cheng et al. derived QOTP-based homomorphic evaluations for additional non-Clifford gates, including controlled square-root operations and Toffoli gates~\cite{cheng2024nonclifford}. These results show that key-dependent QOTP evaluation, non-Clifford gate support, and parametrized rotations are not new in isolation. The distinct focus of QOTPH is their organization into a unified
Pauli frame compiler with explicit direct rules for parametrized controlled rotations, broad gate-level coverage, and a composable correctness specification directly tied to the implementation.

\subsection{Positioning QOTPH: the $T$-gate bottleneck and implementation-level QOTP compilation}
\label{subsec:t-gate-bottleneck-positioning}

A central technical difficulty in QOTP-based quantum homomorphic encryption is the distinction between Clifford and non-Clifford operations. Clifford gates normalize the Pauli group and therefore allow Pauli masks to be propagated by updating a classical key frame. In this case, the same physical Clifford gate can be applied to the encrypted register, while the client updates the QOTP key locally. This property is well understood and forms the algebraic basis of many QOTP-based encrypted-computation schemes.

The situation changes fundamentally for non-Clifford gates. In particular, the $T$ gate does not normalize the Pauli group. Propagating an $X$ mask through a $T$ gate produces an additional $S$-type correction depending on the secret QOTP key. This correction is not a Pauli operator and cannot be absorbed by a simple update of the QOTP bits. Consequently, the treatment of $T$ gates is the main obstruction to extending Clifford key-tracking rules into a fully homomorphic quantum encryption scheme.

Previous theoretical works address this obstruction using additional cryptographic or quantum resources. Broadbent and Jeffery introduced quantum homomorphic encryption schemes whose efficiency depends on the non-Clifford complexity of the evaluated circuit, thereby making explicit the special role of $T$ gates~\cite{broadbent2015quantum}. Dulek, Schaffner, and Speelman constructed a compact QHE scheme for polynomial-size circuits by combining QOTP encryption with classical fully homomorphic encryption and quantum gadgets that handle the non-Pauli corrections arising from $T$-gate evaluation~\cite{dulek2016quantum}. Ouyang, Tan, and Fitzsimons proposed an information-theoretically secure construction based on quantum codes, but for restricted classes of circuits with bounded non-Clifford complexity~\cite{ouyang2018quantum}. More recent generic QFHE constructions, such as that of Gupte and Vaikuntanathan, rely on powerful classical cryptographic primitives and computational assumptions to obtain stronger formal homomorphic capabilities~\cite{gupte2024construct}.

Broadbent-Jeffery and Dulek-Schaffner-Speelman target a cryptographic evaluation model in which key-dependent information must be processed without revealing the underlying QOTP key to the evaluator. Their constructions protect the QOTP key information using classical homomorphic encryption and introduce additional quantum resources to handle the corrections arising from non-Clifford operations~\cite{broadbent2015quantum,dulek2016quantum}. Their primary objective is to obtain stronger cryptographic guarantees and compact homomorphic evaluation, rather than to provide a directly executable hardware-oriented compilation layer for standard gate-based quantum processors.

The works discussed above focus primarily on cryptographic security and compact homomorphic evaluation, rather than on formulating and experimentally validating a unified Pauli-frame compiler that translates a broad library of native, parametrized, controlled, and composite logical gates into a non-adaptive physical circuit directly executable on current gate-based quantum hardware.

QOTPH addresses this implementation and compilation gap. It organizes QOTP key propagation as an explicit client-side Pauli-frame compiler for standard gate-model circuits. The client generates and retains the QOTP frame, resolves the required gate corrections locally, optionally incorporates the initial $X$-mask into computational-basis input data, and produces a non-adaptive physical circuit that can be executed directly using a conventional gate-based quantum backend. The server does not receive the QOTP frame and does not perform key updates or key-dependent decisions during execution. The framework unifies Clifford propagation, one- and two-qubit Pauli-generated rotations, explicit key-dependent rules for parametrized controlled rotations, and recursively compiled composite gates within a single composable correctness specification. To the best of our knowledge, this combination of direct controlled-rotation rules, broad gate-level coverage, client-side compilation, non-adaptive execution, and hardware validation has not previously been presented as a single QOTP compilation framework.

In the delegated workflow, the uncompiled logical circuit and the correspondence between logical and compiled instructions remain on the client side. Consequently, an observed corrected operation does not in general determine a unique QOTP key value without additional knowledge of the logical circuit. This structural ambiguity should not be interpreted as a general circuit-hiding guarantee. Blind quantum computation provides a distinct cryptographic approach to concealing a delegated computation; Universal Blind Quantum Computation (UBQC) is a canonical example~\cite{broadbent2009universal}. Its interactive measurement-based architecture and client requirements differ from the direct non-adaptive gate-based execution targeted by QOTPH, and it is therefore cited here as a complementary reference for the circuit-hiding problem rather than as part of the present construction.

\section{Technical Background}
\label{sec:Backg}

To fully appreciate the scope and contributions of this work, it is essential to first understand the foundational principles that underlie homomorphic encryption and quantum computation. While both areas have undergone substantial development independently, their intersection, quantum homomorphic encryption, remains a frontier with unique challenges and great opportunities. This section consolidates the theoretical foundations necessary to contextualize the proposed quantum homomorphic encryption framework.

\subsection{Homomorphic Encryption}

Homomorphic Encryption (HE) is a form of public-key encryption that allows computations to be performed directly on ciphertexts, in such a way that the result, once decrypted, matches the result of performing the same operations on the plaintext data. This property enables secure delegation of computation to untrusted environments, such as cloud services, without compromising data confidentiality.

Formally, a homomorphic encryption scheme is defined as a tuple of probabilistic polynomial-time algorithms:

\begin{equation}
\text{HE} = (\text{KeyGen}, \text{Enc}, \text{Eval}, \text{Dec}),
\label{eq:HE}
\end{equation}

where:
\begin{itemize}
    \item $ \text{KeyGen}(1^\lambda) \rightarrow (pk, sk) $: generates a public and secret key pair for a given security parameter $ \lambda $.
    \item $ \text{Enc}(pk, m) \rightarrow c $: encrypts a message $ m \in \mathcal{M} $ under the public key $ pk $.
    \item $ \text{Eval}(pk, f, c_1, \ldots, c_n) \rightarrow c_f $: evaluates a function $ f $ over ciphertexts $ c_i = \text{Enc}(pk, m_i) $.
    \item $ \text{Dec}(sk, c_f) \rightarrow f(m_1, \ldots, m_n) $: decrypts the final result, yielding the function $ f $ output on the original plaintexts.
\end{itemize}

A scheme is considered fully homomorphic (FHE) if it supports the evaluation of arbitrary polynomial-size circuits, including both addition and multiplication. The first fully homomorphic encryption construction was proposed by Gentry in 2009~\cite{gentry2009fully}, introducing the concept of bootstrapping to enable the evaluation of arbitrary circuits while preserving ciphertext correctness.

Since then, several optimized variants have been developed, including:
\begin{itemize}
    \item \textbf{BFV}~\cite{fan2012somewhat}: which supports exact arithmetic over integers.
    \item \textbf{BGV}~\cite{brakerski2014leveled}: efficient for leveled schemes over polynomial rings.
    \item \textbf{CKKS}~\cite{cheon2017homomorphic}: which allows approximate arithmetic over real numbers, suitable for machine learning on encrypted data.
    \item \textbf{GSW}~\cite{gentry2013homomorphic}: based on binary matrix encodings, useful for obfuscation and advanced bootstrapping.
\end{itemize}

These schemes fundamentally rely on hard problems from lattice-based cryptography, particularly Learning With Errors (LWE) and Ring-LWE, which are conjectured to be secure against quantum adversaries. The typical underlying ring is 
\begin{equation}
R_q = \mathbb{Z}_q[X]/(X^N + 1),
\end{equation}
with ciphertexts represented as tuples in $R_q^n$, and operations defined over this ring structure.

Homomorphic encryption has enabled applications such as privacy-preserving machine learning, secure multiparty computation, and encrypted database querying. However, the transition to quantum homomorphic encryption, where computations are carried out on encrypted quantum states, remains an open and active research challenge. The following sections study an implementation-oriented QOTP key-tracked compilation framework within this broader landscape.

\subsection{Quantum Computing}

Quantum computing~\cite{nielsen2010quantum} constitutes a non-classical computational model rooted in the principles of quantum mechanics~\cite{shankar2012principles}, offering a radically different approach to information processing. The basic unit of quantum information is the qubit, which exists as a normalized vector in a two-dimensional complex Hilbert space, typically represented in Dirac notation~\cite{dirac1939new} as $\vert{\psi}\rangle = \alpha\vert{0}\rangle + \beta\vert{1}\rangle$ with $|\alpha|^2 + |\beta|^2 = 1$. This allows a coherent superposition of logical states, thereby encoding richer informational content compared to classical bits.

Entanglement plays a central role as a uniquely quantum resource, enabling non-local correlations between qubits. These correlations are formally characterized by the violation of Bell inequalities~\cite{bell1964einstein} and are leveraged in algorithms such as quantum teleportation~\cite{bennett1993teleporting}, superdense coding, and Shor’s factoring algorithm. Importantly, the entangled state space grows exponentially, which underpins the computational parallelism of quantum systems.

A crucial theoretical constraint in quantum information is the no-cloning theorem, which prohibits the perfect copying of arbitrary unknown quantum states. Formally proven by Wootters and Zurek~\cite{wootters1982single} and Dieks~\cite{dieks1982communication}, this theorem imposes structural limits on quantum communication and makes many classical cryptographic primitives, like state replication or duplication-based authentication, fundamentally incompatible with quantum mechanics. However, it also enables intrinsic security advantages in protocols such as quantum key distribution (QKD)~\cite{bennett2014quantum}.

The standard computational model used for algorithm development is the quantum circuit model, where qubits are initialized in known states, manipulated via unitary gates (e.g., Hadamard, $R_z(\theta)$, and CNOT), and measured in the computational basis. This model is supported by the Solovay-Kitaev theorem~\cite{dawson2005solovay}, which ensures the universality of discrete gate sets under suitable approximations.

Recent technological advances, particularly by IBM, Google, and IonQ, have enabled the execution of quantum circuits on real hardware~\cite{arute2019quantum}. Frameworks such as Qiskit~\cite{QiskitT} (by IBM) offer high-level abstractions for quantum programming, transpilation, and simulation. These tools allow for testing of quantum cryptographic protocols, including those involving encrypted computation.

As quantum hardware and software ecosystems continue to mature, these technologies provide a fertile testing ground for experimental implementations of quantum cryptographic primitives. This includes schemes for quantum encryption, authentication, and more ambitiously, quantum homomorphic encryption (QHE), which aims to enable non-trivial operations over encrypted quantum states.

\subsection{Quantum One-Time Pad (QOTP)}
\label{subsec:QOTP}

The Quantum One-Time Pad (QOTP)~\cite{boykin2003optimal} is the quantum analogue of the classical one-time pad~\cite{shannon1949communication}, offering information-theoretic security for quantum data. First formalized by Boykin and Roychowdhury, QOTP provides perfect secrecy for arbitrary quantum states using two classical bits of key per qubit.

Formally, let $\sigma$ be an arbitrary quantum state (pure or mixed). The QOTP encryption of $\sigma$ is given by the following operation:

\begin{equation}
\rho_c = X^j Z^k \, \sigma \, Z^k X^j,
\label{eq:QOTP}
\end{equation}

where $j, k \in \{0,1\}$ are classical bits drawn uniformly at random, and $X$, $Z$ denote the Pauli operators. These operators generate the single-qubit Pauli group and form a group under composition.

The QOTP provides perfect information-theoretic secrecy: when the secret key bits are uniformly random and unknown to the adversary, the ciphertext state averaged over the key distribution is maximally mixed. In particular,

\begin{equation}
\mathbb{E}_{j,k} \left[ X^j Z^k \sigma Z^k X^j \right] = \frac{I}{2},
\end{equation}

for all single-qubit states $\sigma$. For multi-qubit systems, the key length scales linearly: $2n$ classical bits are required to encrypt an $n$-qubit state.

Despite its simplicity and strong security guarantees, the QOTP is limited to symmetric encryption and lacks support for computational manipulation of ciphertexts. In particular, arbitrary quantum gates do not generally commute with Pauli operators, making it challenging to evaluate functions over encrypted quantum data without access to the decryption key.

Nevertheless, due to its simplicity, theoretical rigor, and compatibility with circuit-based quantum computing, QOTP has become a foundational building block in more advanced quantum cryptographic schemes, including protocols for quantum authentication and secure transmission~\cite{ambainis2000private,barnum2002authentication}.

The QOTP, in its standard form, does not support homomorphic computation in the general sense. Arbitrary quantum gates do not generally preserve ciphertext structure under Pauli encryption. In particular, Pauli operators do not commute with all unitary operations, so evaluating general quantum functions on QOTP-encrypted data requires additional key-dependent corrections, auxiliary mechanisms, or an extended homomorphic construction. However, for restricted classes of operations (e.g., Clifford gates), it is possible to analytically track key evolution under encryption, which forms the basis for interactive or assisted computation models~\cite{ouyang2018quantum}.

In this work, we adopt the QOTP as a starting point for constructing a key-tracked client-side compilation framework for encrypted quantum computation. By leveraging its symmetric encryption structure, we explore how to enable specific classes of homomorphic operations while preserving QOTP-based data hiding at the level of the encrypted quantum register and correctness under key-tracked compilation.

For Clifford gates, this key tracking can be performed by updating the Pauli frame. For non-Clifford gates, and in particular for the $T$ gate, Pauli masks are not preserved under conjugation. This obstruction motivates the client-side key-dependent compilation strategy introduced in Section~\ref{sec:Approach}.

\section{Proposed Homomorphic QOTP Scheme}
\label{sec:Approach}

The primary challenge in realizing quantum homomorphic encryption lies in reconciling two seemingly conflicting requirements: the preservation of information-theoretic security for quantum data and the ability to perform non-trivial quantum computations directly on encrypted states. Existing quantum encryption primitives, such as the Quantum One-Time Pad (QOTP), provide strong guarantees of secrecy but are not, in their standard form, compatible with general-purpose encrypted computation due to the non-commutative nature of quantum operations and Pauli operators.

In this section, we present a systematic approach to bridge this gap. Our methodology leverages the structural symmetries of the QOTP and exploits the algebraic properties of quantum gates under Pauli conjugation. By decomposing quantum circuits into elementary gate operations, we enable a homomorphic evaluation model wherein quantum operations are applied to encrypted data, and encryption keys are adaptively updated to track and preserve the semantic equivalence of the computation. Our framework is designed to preserve the QOTP encryption frame throughout the computation, while allowing the desired logical operation to be recovered after decryption. In the construction considered here, all key-dependent gate adaptations are performed locally by the client during circuit compilation, and the resulting circuit is then executed non-interactively by the quantum server. The delegated setting considered throughout the remainder of this work assumes an honest-but-curious evaluator: the server executes the submitted circuit as specified and may inspect the submitted classical circuit description and returned measurement data, while the QOTP frame, the uncompiled logical circuit, and the logical-to-physical instruction mapping remain client-side. Accordingly, the security claims of QOTPH concern the confidentiality of the QOTP-encrypted quantum register and the ambiguity of key-dependent compiled operations under this hidden-logical-circuit model. Circuit hiding, key hiding in the presence of additional side information, integrity, and verifiability are outside the scope of the present construction and are discussed separately in Section~\ref{sec:security}.

\subsection{Foundations of Homomorphic QOTP}

To enable homomorphic quantum computation, we take the QOTP as our cryptographic foundation. The QOTP, when used naively, is homomorphic only with respect to a limited subset of quantum operations, specifically those within the Clifford group, owing to their well-understood commutation relations with Pauli operators. However, supporting broad gate-model computation, including non-Clifford gates, requires a more sophisticated strategy.

Our solution is to analytically characterize how a broad gate set, including Clifford, non-Clifford, controlled, and parametrized operations, transforms QOTP-encrypted states. For each gate, we derive explicit rules for how the encryption keys must be updated to ensure semantic preservation of the computation. Where direct homomorphic application is not possible, we decompose gates into equivalent sequences for which the key-update rules can be tracked. This approach extends QOTP key-tracking ideas to a broad class of gate-model operations and supports non-adaptive remote execution after client-side compilation.

More concretely, let $\rho$ denote an encrypted quantum state under QOTP, specified by key bits $(j,k)$ per qubit, as $\rho = X^j Z^k \sigma Z^k X^j$ for some plaintext state $\sigma$. When a quantum gate $G$ is applied to $\rho$, the result is generally $G\rho G^\dagger$, which may not retain the QOTP encryption structure unless $G$ is compatible with Pauli propagation.

For the classical computational-basis inputs considered in the implementation, the initial $X$-mask can be incorporated directly into the encoded input bits. Specifically, the client may prepare $m_i\oplus j_i$ instead of $m_i$, so that the explicit initial $X^{j_i}$ gates need not appear in the submitted circuit. The initial $Z$-mask contributes only an irrelevant global phase to a computational-basis input and may therefore be retained only in the client-side QOTP frame when this optimization is used. 

As an alternative preparation mechanism supported by QOTPH, the client may construct a temporary local circuit containing the input preparation and the QOTP $X$- and $Z$-mask operations, extract the resulting encrypted statevector, discard the temporary circuit, and initialize a fresh execution circuit from the encrypted state before appending the homomorphically compiled operations. Only this fresh circuit is submitted to the server, so the original input-encoding and QOTP-encryption gates are not present as separate identifiable layers in the submitted circuit.

For Clifford gates, this compatibility allows the server to apply the same logical gate while the client updates the QOTP key frame locally. For non-Clifford gates and parametrized rotations, the physical gate applied in the compiled circuit may depend on the current QOTP key. In the present construction, these key-dependent adaptations are computed locally by the client before execution. The server then executes the submitted circuit non-adaptively and returns the corresponding measurement outcomes.

This methodology provides a systematic prescription for implementing QOTP key-tracked encrypted computation over Clifford, non-Clifford, controlled, and parametrized operations. The construction should be understood as a client-side compilation framework for non-interactive remote execution, with data privacy separated from full privacy of the submitted circuit description.

\subsubsection{Clifford propagation and the non-Clifford obstruction}
\label{subsec:clifford-nonclifford-obstruction}

A central distinction in QOTP-based homomorphic evaluation is the difference between Clifford and non-Clifford operations. Let a single-qubit QOTP encryption be written as
\begin{equation}
\mathrm{Enc}_{j,k}(\sigma)
=
X^{j} Z^{k} \sigma Z^{k} X^{j},
\label{eq:qotp-single-qubit-encryption}
\end{equation}
where $j,k\in\{0,1\}$. For any Clifford gate $C$, the Pauli group is normalized by conjugation. Therefore, for every pair $(j,k)$, there exist updated bits $(j',k')$ and an irrelevant global phase $\omega$ such that
\begin{equation}
C X^{j} Z^{k}
=
\omega X^{j'} Z^{k'} C.
\label{eq:clifford-pauli-propagation}
\end{equation}
Consequently, the server may apply the same logical Clifford gate $C$ to the encrypted register, while the client only updates the classical QOTP key frame locally.

This property no longer holds for non-Clifford gates. In particular, the $T$ gate does not normalize the Pauli group. While
\begin{equation}
T Z T^{\dagger}
=
Z,
\label{eq:t-z-t}
\end{equation}
one has
\begin{equation}
T X T^{\dagger}
=
e^{-i\pi/4} S X.
\label{eq:t-x-t}
\end{equation}
Thus, the propagation of an $X$-mask through a $T$ gate produces an additional $S$-type correction depending on the secret QOTP bit $j$. This correction is not a Pauli mask and therefore cannot be absorbed by updating only the QOTP bits $(j,k)$. This is the standard algebraic obstruction that makes non-Clifford gates the main difficulty in QOTP-based quantum homomorphic encryption.

In the present work, this obstruction is handled by client-side key-dependent compilation. The quantum server is not asked to compute the correction, update the key frame, or choose gates based on the secret key. Instead, the client constructs the physical circuit locally. For $R_z$-type rotations, the client uses the identity
\begin{equation}
R_z\left((-1)^j\theta\right) X^{j} Z^{k}
=
X^{j} Z^{k} R_z(\theta),
\label{eq:rz-key-dependent-correction}
\end{equation}
up to a global phase, and therefore replaces the logical rotation $R_z(\theta)$ by the corrected physical rotation $R_z(\theta')$, with
\begin{equation}
\theta'
=
(-1)^j\theta.
\label{eq:rz-corrected-angle}
\end{equation}
The $T$ and $T^{\dagger}$ gates are then treated as the special cases $\theta=\pi/4$ and $\theta=-\pi/4$, respectively, up to global phases.

This makes the execution phase non-interactive: once the circuit has been locally compiled by the client, the server only executes the submitted gate sequence and returns measurement outcomes. The resulting classical circuit description may contain key-dependent corrected angles, whose interpretation within the evaluator model is analyzed in Section~\ref{sec:security}.

\subsection{Homomorphic Gate Decompositions}

The practical realization of this approach relies on explicit gate-by-gate rules for homomorphic evaluation and key update. Table~\ref{homomorphic-decompositions} summarizes the homomorphic implementation of a broad set of elementary and controlled gates. For each operation, we indicate either a direct homomorphic equivalent, a decomposition into basic gates compatible with key tracking, or a specific update prescription for the QOTP keys. 

These rules collectively form the operational backbone of our homomorphic evaluation protocol, maintaining the QOTP key frame throughout the computation and enabling correctness validation after decryption. Notably, this approach covers a broad gate set, including Clifford, non-Clifford, controlled, and parametrized operations, enabling the experimental study of a wide range of gate-based circuits. The explicit key update rules allow the remote execution phase to remain non-interactive, since all key-dependent adaptations are performed locally by the client during circuit compilation. The server is therefore not required to update keys, choose corrected gates, or perform adaptive operations during execution.

Importantly, Table~\ref{homomorphic-decompositions} should be interpreted as a client-side compilation table. Entries involving corrected angles or conditional decompositions are computed by the client using the secret QOTP key frame before the circuit is submitted to the quantum server.

In Table~\ref{homomorphic-decompositions}, gate sequences are listed in chronological execution order from left to right: the leftmost gate is appended and executed first. In operator identities and formal proofs, the standard convention in which the rightmost operator acts first is used.

\begin{table}[h!]

\centering

\caption{\bf Client-side QOTPH gate decompositions and key-frame updates}

\small

\begin{tabular}{p{1.5cm}p{5.2cm}}

\hline

Gate & Client-side compiled implementation \\

\hline
$X$ & $X$, update: none \\
$Y$ & $Y$, update: none \\
$Z$ & $Z$, update: none \\
$H$ & $H$, update: $(j,k)\gets(k,j)$ \\
$R_x(\theta)$ & $H\cdot R_z(\theta')\cdot H$, where $\theta'=(-1)^k\theta$ \\
$R_y(\theta)$ & $R_y(\theta')$, where $\theta'=(-1)^{j\oplus k}\theta$ \\
$R_z(\theta)$ & $R_z(\theta')$, where $\theta'=(-1)^j\theta$ \\
$S$ & $S$, update: $k\gets k\oplus j$ \\
$S^\dagger$ & $R_z(\theta')$, where $\theta'=(-1)^j(-\pi/2)$ \\
$T$ & $R_z(\theta')$, where $\theta'=(-1)^j(\pi/4)$ \\
$T^\dagger$ & $R_z(\theta')$, where $\theta'=(-1)^j(-\pi/4)$ \\
$\mathrm{CX}$ & $\mathrm{CX}$, update: $j_t\gets j_t\oplus j_c$, $k_c\gets k_c\oplus k_t$ \\
$\mathrm{CY}$ & $\mathrm{CY}$, frame update from $S^\dagger\cdot\mathrm{CX}\cdot S$ \\
$\mathrm{CZ}$ & $\mathrm{CZ}$, update: $k_c\gets k_c\oplus j_t$, $k_t\gets k_t\oplus j_c$ \\
$\mathrm{CH}$ & $R_y(\pi/4)\cdot \mathrm{CX}\cdot R_y(-\pi/4)$, update: each gate update \\
$\mathrm{CS}$ & $T^{(c)}\cdot T^{(t)}\cdot \mathrm{CX}_{c,t}\cdot (T^\dagger)^{(t)}\cdot \mathrm{CX}_{c,t}$, update: each gate update; the net final key frame is unchanged \\
$CR_z(\pi/4)$ & $R_z^{(t)}(\pi/8)\cdot \mathrm{CX}_{c,t}\cdot R_z^{(t)}(-\pi/8)\cdot \mathrm{CX}_{c,t}$, update: each gate update \\
$CR_x(\theta)$ & if $j_c=0$: $CR_x(\theta')$; if $j_c=1$: $R_x(\theta')\cdot CR_x(-\theta')$ on $(c,t)$, where $\theta'=(-1)^{k_t}\theta$ \\
$CR_y(\theta)$ & if $j_c=0$: $CR_y(\theta')$; if $j_c=1$: $R_y(\theta')\cdot CR_y(-\theta')$ on $(c,t)$, where $\theta'=(-1)^{j_t\oplus k_t}\theta$ \\
$CR_z(\theta)$ & if $j_c=0$: $CR_z(\theta')$; if $j_c=1$: $R_z(\theta')\cdot CR_z(-\theta')$ on $(c,t)$, where $\theta'=(-1)^{j_t}\theta$ \\
$RXX(\theta)$ & $RXX(\theta')$ on $(p,q)$, where $\theta'=(-1)^{k_p\oplus k_q}\theta$ \\
$RYY(\theta)$ & $RYY(\theta')$ on $(p,q)$, where $\theta'=(-1)^{j_p\oplus k_p\oplus j_q\oplus k_q}\theta$ \\
$RZZ(\theta)$ & $RZZ(\theta')$ on $(p,q)$, where $\theta'=(-1)^{j_p\oplus j_q}\theta$ \\
$\mathrm{CXX}$ & $\mathrm{CX}_{c,t_1}\cdot \mathrm{CX}_{c,t_2}$, update: each gate update \\
$\sqrt{X}$ & $\sqrt{X}$, update: $j\gets j\oplus k$ \\
$\mathrm{SWAP}$ & $\mathrm{SWAP}$, exchange: $K_1\leftrightarrow K_2$ \\
$\mathrm{CCX}$ & $R_z^{(c)}(\pi/2)\cdot \sqrt{X}^{(c)}\cdot R_z^{(c)}(\pi/2)\cdot \mathrm{CX}(b,c)\cdot R_z^{(c)}(7\pi/4)\cdot \mathrm{CX}(a,c)\cdot R_z^{(c)}(\pi/4)\cdot \mathrm{CX}(b,c)\cdot R_z^{(c)}(7\pi/4)\cdot \mathrm{CX}(a,c)\cdot R_z^{(b)}(\pi/4)\cdot R_z^{(c)}(\pi/4)\cdot R_z^{(c)}(\pi/2)\cdot \sqrt{X}^{(c)}\cdot R_z^{(c)}(\pi/2)\cdot \mathrm{CX}(a,b)\cdot R_z^{(a)}(\pi/4)\cdot R_z^{(b)}(7\pi/4)\cdot \mathrm{CX}(a,b)$, update: each gate update \\
$\mathrm{CCZ}$ & $H\cdot \mathrm{CCX}\cdot H$, update: each gate update \\
Fredkin & $\mathrm{CX}_{t_2,t_1}\cdot \mathrm{CCX}_{c,t_1,t_2}\cdot \mathrm{CX}_{t_2,t_1}$, update: each gate update \\
\hline
\end{tabular}
\label{homomorphic-decompositions}
\end{table}

This methodology preserves the QOTP encryption frame throughout the compiled computation and facilitates automated tracking of the encryption keys during circuit compilation. As we show in the following sections, these rules are readily translatable into code, enabling efficient realization of quantum homomorphic evaluation.

The server is not given the logical angle $\theta$, the uncompiled logical circuit, or the correspondence between logical and compiled instructions. Consequently, an observed gate $R_z(\alpha)$ is consistent both with a logical gate $R_z(\alpha)$ under $j=0$ and with a logical gate $R_z(-\alpha)$ under $j=1$. Likewise, a physical gate subsequence may originate from the logical circuit itself or from a compiler-introduced correction or decomposition. The compiled description therefore does not directly identify a unique QOTP key.
This is a non-unique-inversion property under the hidden-logical-circuit model, rather than a general cryptographic indistinguishability guarantee. If the exact logical circuit, including the relevant signed logical parameters and the correspondence between logical and compiled instructions, is known, corrected rotation signs and conditional decompositions may reveal key-frame bits or parities. Knowledge of the algorithmic family alone is generally weaker because instance-dependent or variational parameter signs may remain unknown; this distinction is analyzed in Section~\ref{subsec:algorithm-aware-analysis}.

\subsection{Formal Correctness of the QOTPH Compiler}
\label{subsec:formal-correctness}

For an $n$-qubit QOTP frame
$K=(\mathbf{j},\mathbf{k})\in\mathbb{F}_2^{2n}$, define
\begin{equation}
P_K
=
\bigotimes_{i=1}^{n} X_i^{j_i} Z_i^{k_i}.
\label{eq:formal-pauli-frame}
\end{equation}
Let the logical circuit be
\begin{equation}
\mathcal{C}
=
G_L G_{L-1}\cdots G_1,
\label{eq:formal-logical-circuit}
\end{equation}
where $G_1$ acts first. Starting from the initial frame $K_0$, the client-side compiler processes each logical gate according to
\begin{equation}
\left(
\widetilde{G}_{\ell},K_{\ell}
\right)
=
\operatorname{CompileGate}
\left(
G_{\ell},K_{\ell-1}
\right).
\label{eq:formal-compile-gate}
\end{equation}

Let $\mathcal{G}_{\mathrm{ver}}$ denote the verified primitive gate set described in Appendix~\ref{app:correctness-proofs}. A logical gate may also be supported through an exact decomposition into gates from $\mathcal{G}_{\mathrm{ver}}$. In either case, the compiled operation $\widetilde{G}_{\ell}$ and the updated frame $K_{\ell}$ must satisfy
\begin{equation}
\widetilde{G}_{\ell}P_{K_{\ell-1}}
=
e^{i\phi_{\ell}}
P_{K_{\ell}}G_{\ell},
\label{eq:local-compilation-correctness}
\end{equation}
where $\phi_{\ell}\in\mathbb{R}$ is a physically irrelevant global phase.

\begin{theorem}[Correctness of QOTPH compilation]
\label{thm:qotph-correctness}
Let
\begin{equation}
\widetilde{\mathcal{C}}_{K_0}
=
\widetilde{G}_L
\widetilde{G}_{L-1}
\cdots
\widetilde{G}_1
\end{equation}
be the physical circuit produced from $\mathcal{C}=G_L\cdots G_1$ and the initial QOTP frame $K_0$.
Assume that, for every $\ell\in\{1,\ldots,L\}$, the corresponding compiled operation and updated frame satisfy Eq.~\eqref{eq:local-compilation-correctness}. Then
\begin{equation}
\widetilde{\mathcal{C}}_{K_0}P_{K_0}
=
e^{i\Phi}P_{K_L}\mathcal{C},
\qquad
\Phi
=
\sum_{\ell=1}^{L}\phi_{\ell}.
\label{eq:global-compilation-correctness}
\end{equation}
Consequently, for every input density operator $\rho$,
\begin{equation}
\begin{aligned}
&
\widetilde{\mathcal{C}}_{K_0}
P_{K_0}\rho P_{K_0}^{\dagger}
\widetilde{\mathcal{C}}_{K_0}^{\dagger}
 =
P_{K_L}
\mathcal{C}\rho\mathcal{C}^{\dagger}
P_{K_L}^{\dagger}.
\end{aligned}
\label{eq:global-density-correctness}
\end{equation}
Therefore, applying $P_{K_L}^{\dagger}$ to the final encrypted register recovers the logical output $\mathcal{C}\rho\mathcal{C}^{\dagger}$.
\end{theorem}

\begin{proof}
The proof proceeds by induction over the circuit length. For the first gate, Eq.~\eqref{eq:local-compilation-correctness} gives
\begin{equation}
\widetilde{G}_1P_{K_0}
=
e^{i\phi_1}P_{K_1}G_1.
\end{equation}
Assume that the identity holds after the first $m$ gates, with
\begin{equation}
\Phi_m
=
\sum_{\ell=1}^{m}\phi_\ell,
\end{equation}
namely,
\begin{equation}
\widetilde{G}_m\cdots\widetilde{G}_1P_{K_0}
=
e^{i\Phi_m}
P_{K_m}G_m\cdots G_1.
\label{eq:qotph-induction-hypothesis}
\end{equation}
Multiplying Eq.~\eqref{eq:qotph-induction-hypothesis} from the
left by $\widetilde{G}_{m+1}$ gives
\begin{align}
&
\widetilde{G}_{m+1}
\widetilde{G}_m\cdots
\widetilde{G}_1P_{K_0}
\nonumber\\
&\quad =
e^{i\Phi_m}
\widetilde{G}_{m+1}
P_{K_m}
G_m\cdots G_1
\nonumber\\
&\quad =
e^{i(\Phi_m+\phi_{m+1})}
P_{K_{m+1}}
G_{m+1}G_m\cdots G_1,
\end{align}
where the final equality follows from the local identity for $G_{m+1}$. The operator identity therefore holds for every circuit length $L$.

Applying the operator identity to $\rho$ and multiplying by its adjoint gives Eq.~\eqref{eq:global-density-correctness}. The global phase cancels from the density operator.
\end{proof}

Write the final QOTP frame as
\begin{equation}
K_L
=
\left(
\mathbf{j}^{\mathrm{final}},
\mathbf{k}^{\mathrm{final}}
\right).
\end{equation}

For computational-basis measurements, if the server returns an encrypted bitstring $x$, local decryption only requires the final $X$-mask:
\begin{equation}
y
=
x\oplus\mathbf{j}^{\mathrm{final}}.
\label{eq:classical-output-decryption}
\end{equation}
The final $Z$-mask does not modify computational-basis measurement probabilities.

The XOR in Eq.~\eqref{eq:classical-output-decryption} is applied componentwise after expressing the measured bitstring and the final QOTP frame in the same qubit-order convention. Any additional final qubit permutation must be inverted separately.

Theorem~\ref{thm:qotph-correctness} reduces the correctness of the complete compiled circuit to the verification of the corresponding single-gate identities. For Clifford gates, these identities follow from Pauli conjugation and produce classical key-frame updates. For Pauli-generated rotations, the client corrects the physical rotation
angle according to the commutation parity of the current QOTP frame.
Controlled rotations are handled through their key-dependent compilation rules, whereas higher-level operations are reduced to exact decompositions into verified primitive gates. The following CNOT derivation illustrates the Clifford case explicitly, while the remaining primitive identities and the composition argument are provided in Appendix~\ref{app:correctness-proofs}.

\subsection{QOTP key propagation through CNOT}
\label{subsec:cnot-key-propagation-example}

As a practical example of how QOTP key-frame update rules are obtained, we derive the propagation rule for a CNOT gate with control qubit $c$ and target qubit $t$. This example illustrates the general method used for Clifford gates: the physical gate applied to the encrypted register is the same logical gate, while the client updates the classical QOTP key frame locally.

Let the two-qubit QOTP mask before the CNOT gate be
\begin{equation}
X_c^{j_c} Z_c^{k_c} \otimes X_t^{j_t} Z_t^{k_t},
\end{equation}
where $(j_c,k_c)$ and $(j_t,k_t)$ are the QOTP key bits associated with the control and target qubits, respectively. Equivalently, this mask can be written as a product of embedded single-qubit Pauli operators:
\begin{equation}
\begin{aligned}
X_c^{j_c} Z_c^{k_c} \otimes X_t^{j_t} Z_t^{k_t}
&=
(X_c\otimes I_t)^{j_c}
(Z_c\otimes I_t)^{k_c}
\\
&\quad \times
(I_c\otimes X_t)^{j_t}
(I_c\otimes Z_t)^{k_t}.
\end{aligned}
\end{equation}

Using the standard CNOT conjugation identities,
\begin{equation}
\begin{aligned}
\mathrm{CX}_{c,t}(X_c\otimes I_t)\mathrm{CX}_{c,t} &= X_c\otimes X_t,\\
\mathrm{CX}_{c,t}(Z_c\otimes I_t)\mathrm{CX}_{c,t} &= Z_c\otimes I_t,\\
\mathrm{CX}_{c,t}(I_c\otimes X_t)\mathrm{CX}_{c,t} &= I_c\otimes X_t,\\
\mathrm{CX}_{c,t}(I_c\otimes Z_t)\mathrm{CX}_{c,t} &= Z_c\otimes Z_t,
\end{aligned}
\end{equation}

we obtain, up to an irrelevant global phase,
\begin{equation}
\begin{aligned}
&\mathrm{CX}_{c,t}
\left(
X_c^{j_c}Z_c^{k_c}\otimes X_t^{j_t}Z_t^{k_t}
\right)
\mathrm{CX}_{c,t}
\\
&\qquad =
\pm
\left(
X_c^{j_c}Z_c^{k_c\oplus k_t}
\otimes
X_t^{j_c\oplus j_t}Z_t^{k_t}
\right).
\end{aligned}
\end{equation}

Therefore, after applying a CNOT gate to the encrypted register, the updated QOTP key frame is
\begin{equation}
\begin{aligned}
j_c' &= j_c,\\
k_c' &= k_c\oplus k_t,\\
j_t' &= j_t\oplus j_c,\\
k_t' &= k_t.
\end{aligned}
\end{equation}
Equivalently, the key update can be written as
\begin{equation}
j_t \gets j_t\oplus j_c,
\qquad
k_c \gets k_c\oplus k_t,
\end{equation}
with $j_c$ and $k_t$ unchanged.

This derivation illustrates why Clifford gates are compatible with QOTP key tracking: the CNOT gate maps Pauli masks to Pauli masks under conjugation, so the encryption frame can be updated purely classically by the client. More complex Clifford decompositions are handled by composing the corresponding primitive update rules gate by gate. Non-Clifford gates, in contrast, require the key-dependent compilation procedure described above.

\subsection{The Protocol}

Specifically, the protocol proceeds as follows:  
Given a set of classical input bits $(m_1, ..., m_n)$, a quantum circuit $\mathcal{C}$ (defined as an ordered list of gates), and a desired number of executions $N$, the implementation first prepares the quantum state $\vert{\psi}\rangle = \vert{m_1 \ldots m_n}\rangle$ by initializing the circuit accordingly. 

Subsequently, for each qubit, two uniformly random bits $(a_i,b_i)$ are generated and the corresponding $X^{a_i}Z^{b_i}$ operators are applied; this constitutes the quantum one-time pad (QOTP) encryption. The encrypted state is then processed gate by gate according to $\mathcal{C}$. For each logical gate, the implementation inserts its compiled encrypted-domain version or decomposition, as specified in Table~\ref{homomorphic-decompositions}, and updates the QOTP frame according to the corresponding algebraic rule.

After homomorphic evaluation, QOTPH supports two result-handling workflows. In the in-circuit validation workflow, the client appends the final QOTP decryption layer before submitting the circuit, and the server executes the complete circuit and returns plaintext measurement outcomes. This variant is used only for correctness validation and exposes the final key-dependent decryption layer to the server.

In the delegated local-decryption workflow, the final decryption layer is omitted from the submitted circuit. The server measures the encrypted output register and returns encrypted outcome counts. The client then decrypts the outcomes locally using the final QOTP frame $K_L$.

The protocol proceeds step by step as follows:

\begin{enumerate}
    \item The protocol involves two entities: a client and a quantum server. The client retains the QOTP key frame and performs all key-dependent compilation locally, while the server receives the circuit submitted for execution.
    \item Local key generation: The client locally generates a secret key $K$. This key remains exclusively on the client side throughout the protocol.
    \item The client designs the quantum circuit to be executed on the server. This circuit typically includes the encoded input data, the encryption procedure, and the quantum operations to be applied.
    \item Quantum circuit preparation: During circuit construction, after encoding the input data and applying the encryption, when entering the operations, each quantum gate $G_\ell$ is replaced by a corrected gate $G'_\ell$. This corrected gate is determined by the encryption key computed by the client. The corresponding decryption key is also updated for each adjusted gate introduced. At no point does the server need to modify the gates or access the keys; the client locally constructs the complete homomorphic quantum circuit to be executed by the server.
    \item The client sends the prepared quantum circuit to the quantum provider for execution. Before submission, the client transpiles the circuit to the backend basis and removes logical-gate labels, barriers, compiler annotations, key-update records, and logical-to-physical instruction metadata.
    \item The quantum server returns the measurement results obtained from executing the submitted circuit.
    \item The client locally decrypts the returned results using the decryption key determined during circuit construction, thereby obtaining a histogram of the decrypted outcomes.
\end{enumerate}

The submitted circuit is a compiled physical representation of the client-side homomorphic construction. In particular, the server receives only the circuit to be executed and does not receive the QOTP keys, does not update the key frame, and does not choose key-dependent corrections during execution. Depending on the backend and compilation pipeline, the submitted circuit may undergo further transpilation before execution; however, we do not rely on transpilation as a cryptographic hiding mechanism.

In this work, the quantum provider is modeled as an honest-but-curious evaluator: it is assumed to execute the submitted circuit as specified, while potentially having access to the classical circuit description and to the returned measurement data. Under this model, the protocol provides QOTP-based protection of the encrypted quantum data and, in the local-decryption version, of the measured outcomes. Full privacy of the submitted circuit description is treated as a separate requirement that can be addressed by additional circuit-hiding layers, while integrity and verifiability can be addressed by separate verification mechanisms; these extensions are outside the scope of the basic QOTPH construction.

\section{Implementation of the Proposed Scheme}
\label{sec:Implementation}

In this section, we detail the implementation of the proposed QOTP-based key-tracked client-side compilation framework. The implementation constructs encrypted quantum circuits under the QOTP scheme, applies the corresponding client-side gate corrections and key-frame updates, and evaluates the resulting compiled circuits on quantum simulators and gate-based hardware.

\subsection{Experimental Setup}

The proposed procedure is implemented programmatically as a sequence of Qiskit circuit constructions and manipulations, where encryption, gate application, key tracking, and decryption are carried out in strict accordance with the protocol described in the pseudocode included below.

The implementation is divided into two main steps. The first, Quantum Homomorphic Evaluation Core, is responsible for managing the QOTP encryption, applying the gate decompositions along with the corresponding key updates, and executing the homomorphic evaluation of the circuit in a non-interactive manner. The second one, Decryption and Result Extraction, performs the final decryption either within the circuit for immediate result analysis or via post-processing of the encrypted outputs for local decryption.

The experiments were performed using Qiskit version 2.1.1 and the IBM Quantum backend \texttt{ibm\_basquecountry}. Circuits were transpiled using Qiskit's preset pass manager at optimization level 3, targeting the selected backend, before submission with SamplerV2. Backend calibration parameters vary over time; the reported results therefore correspond to the device conditions present at the time of execution.

\subsection{Algorithmic Specification}

The implemented and fully functional version of the proposed protocol is formalized in Algorithm~\ref{alg:qhe}. The algorithm explicitly describes the complete operational workflow of the homomorphic protocol, from encrypted state preparation to in-circuit decryption and measurement. The implementation is designed as a modular and extensible Qiskit-based system that is agnostic to the specific quantum backend. 

The implemented protocol is organized into three conceptually distinct stages. First, the client performs all QOTP encryption, key-frame tracking, and key-dependent gate compilation locally. This stage produces a compiled physical circuit whose gates already include the required QOTP-dependent corrections. Second, the quantum server receives only this compiled circuit and executes it non-adaptively, without access to the QOTP keys and without performing any key update. Third, the client decrypts the returned encrypted measurement outcomes locally using the final QOTP key frame.

This separation is important for the security model considered in this work. All key-dependent transformations are computed on the client side before execution, while the server-side task is limited to executing the submitted circuit and returning measurement statistics.

Throughout the protocol, each QOTP key is denoted by
\[
K_i=(j_i,k_i),
\]
where \(j_i\) represents the \(X\)-mask and \(k_i\) represents the \(Z\)-mask acting on qubit \(i\).
The encrypted state is therefore evaluated gate by gate while the key frame \(K=\{K_i\}_{i=1}^n\) is dynamically tracked.

We write the circuit as an ordered gate sequence
\[
\mathcal{C}=\bigl\{(G_\ell,Q_\ell,\Theta_\ell)\bigr\}_{\ell=1}^{L},
\]
where \(G_\ell\) denotes the \(\ell\)-th logical gate, \(Q_\ell\) is the ordered set of qubits on which it acts, and \(\Theta_\ell\) contains its continuous parameters, if any.

Algorithm~\ref{alg:qhe} is primarily used as an in-circuit decryption variant for correctness validation. In an untrusted delegated setting, the final QOTP decryption layer should not be submitted to the server, since it depends on the secret key frame. The delegated encrypted-output variant is therefore the local-decryption workflow described in Algorithm~\ref{alg:qhe-local-circuit-decrypt}.

\begin{breakablealgorithm}

\caption{QOTPH quantum homomorphic evaluation protocol with client-side compilation and in-circuit decryption}
\small
\label{alg:qhe}
\begin{algorithmic}[1]
\Procedure{HomomorphicEvaluate}{$m_1,\ldots,m_n$ classical bits, $\mathcal{C}$ gate sequence, $N$ shots}
\Statex \textbf{Execution-domain convention:}
\Statex Steps 1-4 are performed by the client during local circuit construction and compilation.
\Statex Step 5 is performed by the external quantum server after receiving the compiled circuit.
\Statex The QOTP key frame $K$ is generated, updated, and stored only on the client side.
\Statex The quantum server does not receive $K$ and does not perform key-frame updates.
\State{\bf{1. State Preparation}}
\State Prepare $\vert{\psi}\rangle \gets \vert{m_1 \ldots m_n}\rangle$
\Comment{Initialize quantum state from classical input data}

\State{\bf{2. QOTP Encryption}}
\For{$i = 1$ to $n$}
\State $(j_i,k_i) \gets$ uniformly random bits in $\{0,1\}^2$
\State Apply $X^{j_i}Z^{k_i}$ to qubit $i$ of $\vert{\psi}\rangle$
\EndFor
\State $K \gets \{(j_i,k_i)\}_{i=1}^n$
\State $\vert{\psi}_{\mathrm{enc}}\rangle \gets \vert{\psi}\rangle$

\State{\bf{3. Homomorphic Gate Evaluation -- client-side compilation}}
\For{$\ell = 1$ to $L$}
  \State Let $(G_\ell,Q_\ell,\Theta_\ell)$ be the $\ell$-th instruction of $\mathcal{C}$
  \State $K \gets \Call{EvaluateGate}{G_\ell,Q_\ell,\Theta_\ell,K}$
\EndFor

\State{\bf{4. QOTP Decryption}}
\For{$i = 1$ to $n$}
\State $(j_i,k_i) \gets K_i$
\State Apply $Z^{k_i}X^{j_i}$ to qubit $i$ of $\vert{\psi}_{\mathrm{enc}}\rangle$
\EndFor
\State $\vert{\psi}_{\mathrm{out}}\rangle \gets \vert{\psi}_{\mathrm{enc}}\rangle$

\State{\bf{5. Quantum Execution and Result Extraction -- server-side execution}}
\State Measure all qubits of $\vert{\psi}_{\mathrm{out}}\rangle$ and repeat the execution $N$ times
\State Collect measurement statistics $\{p(x):x\in\{0,1\}^n\}$
\State \Return measurement statistics

\EndProcedure
\end{algorithmic}
\end{breakablealgorithm}

For the delegated computational-basis workflow, the client prepares the masked bitstring $\vert{\mathbf{m}\oplus\mathbf{j}}\rangle$ directly, rather than submitting a separate visible layer of initial $X$-mask gates. The initial $Z$-mask is retained only in the client-side frame, since it contributes a global phase to the initial basis state. The explicit QOTP gate layer shown in Algorithm~\ref{alg:qhe} is used for correctness validation.

The homomorphic evaluation step is implemented by a client-side gate-compilation routine that applies the rules summarized in Table~\ref{homomorphic-decompositions}. Since the full routine is lengthy and mainly serves to specify implementation details, the complete pseudocode is provided in Appendix~\ref{app:gate-compilation-rules}. This appendix includes the full gate-dispatch routine, the primitive QOTP key-frame update rules, and the composite-gate decompositions used by the implementation.

In the main text, we emphasize the conceptual structure: all key-dependent angle corrections, decompositions, and key-frame updates are computed locally by the client before the compiled circuit is submitted to the quantum server. Each logical operation is represented as \((G_\ell,Q_\ell,\Theta_\ell)\), where \(G_\ell\) is the gate, \(Q_\ell\) its ordered qubit support, and \(\Theta_\ell\) its parameters when applicable. Depending on the gate, the encrypted-domain transformation is performed through a key-frame update, an angle correction, or an explicit decomposition. Pure parametrized rotations are evaluated by correcting their effective angles according to the current encryption mask and do not introduce an independent key-frame update. When a parametrized operation is implemented through a decomposition, the intermediate elementary gates are handled recursively by the same client-side compilation rules.

A key feature of this implementation is the integration of the key-update logic into the client-side compilation layer. At each logical gate, the corresponding transformation rule is applied locally by the client, so that the QOTP encryption frame is consistently maintained while the compiled physical circuit is constructed. This allows the user to specify a circuit at the logical level without manually compensating for the encryption masks, since all angle corrections, key updates, and gate decompositions are internally handled by the evaluation engine.

The external quantum server does not run the key-update logic and does not receive the QOTP key frame. Its role is limited to executing the compiled circuit submitted by the client and returning measurement statistics. Thus, from the user's perspective, the implementation remains automated, while the distinction between client-side compilation and server-side execution is kept explicit.

The modular design also facilitates adaptation to new circuit architectures, alternative encryption schemes, or different quantum hardware backends. Extending the framework only requires adding new transformation rules or decomposition routines to the evaluation layer. Encryption, homomorphic evaluation, decryption, and measurement are therefore orchestrated as a single workflow, requiring no user intervention once the input circuit has been specified. This architecture enables automated experimentation and correctness validation for QOTP-based encrypted computation.

\subsection{Encrypted Result Handling and Local Decryption}

To enable local decryption of the results, we introduce a second workflow in which measurement outcomes are returned in encrypted form and decrypted by the client after execution. For this purpose, the results are collected while still encrypted; that is, the decryption step is intentionally omitted from the circuit executed on the quantum device or simulator. Using either a quantum simulator or an actual quantum processor, these encrypted measurement outcomes can be retrieved and subsequently decrypted locally by the client. 

Algorithm~\ref{alg:qhe-local-circuit-decrypt} describes this local-decryption workflow. For each encrypted bitstring outcome and its frequency count, the client constructs a local decryption circuit that prepares the corresponding computational-basis state.

\begin{breakablealgorithm}

\caption{QOTPH Quantum Homomorphic Evaluation Protocol with Local Decryption}
\label{alg:qhe-local-circuit-decrypt}
\begin{algorithmic}[1]
\small
\Procedure{HomomorphicEvaluate}{$m_1, ..., m_n$ (classical bits), $\mathcal{C}$ (gate sequence), $N$ (shots)}

\Statex \textbf{Execution-domain convention:}
\Statex Steps 1-3 are performed by the client during local circuit construction and compilation.
\Statex Step 4 is performed by the external quantum server after receiving the compiled circuit.
\Statex The QOTP key frame $K$ is generated, updated, and stored only on the client side.
\Statex The quantum server does not receive $K$ and does not perform key-frame updates.

  \State{\bf{1. State Preparation}}
  \State Prepare $\vert{\psi}\rangle \gets \vert{m_1 \ldots m_n}\rangle$ \Comment{Initialize quantum state from data}
  \State{\bf{2. QOTP Encryption}}
  \For{$i = 1$ to $n$}
    \State $(a_i, b_i) \gets$ random bits in $\{0,1\}^2$
    \State Apply $X^{a_i} Z^{b_i}$ to qubit $i$ of $\vert{\psi}\rangle$
  \EndFor
  \State $K \gets \{(a_i, b_i)\}_{i=1}^n$
  \State $\vert{\psi}_{\mathrm{enc}}\rangle \gets \vert{\psi}\rangle$
\State{\bf{3. Homomorphic Gate Evaluation -- client-side compilation}}
\For{$\ell = 1$ to $L$}
  \State Let $(G_\ell,Q_\ell,\Theta_\ell)$ be the $\ell$-th instruction of $\mathcal{C}$
  \State $K \gets \Call{EvaluateGate}{G_\ell,Q_\ell,\Theta_\ell,K}$
\EndFor
  \State{\bf{4. Quantum Execution and Result Extraction -- server-side execution}}
  \State Measure all qubits of the final encrypted register $\vert{\psi}_{\mathrm{enc,out}}\rangle$, repeat $N$ times
  \State Obtain encrypted measurement statistics $\{p(x) : x \in \{0,1\}^n\}$
  \State Send encrypted measurement statistics $\{p(x):x\in\{0,1\}^n\}$ to the client
  \State {\bf 5. Local Quantum Circuit QOTP Decryption -- client-side post-processing}
  \For{each bitstring $x$ with nonzero count}
    \State $qc_{\mathrm{dec}} \gets$ New quantum circuit of $n$ qubits
    \State Prepare state $\vert{x}\rangle$ in $qc_{\mathrm{dec}}$
    \Statex \Comment{Local decryption uses the final QOTP key frame after all homomorphic key updates}
    \For{$i = 1$ to $n$}
      \State $(a_i^{\mathrm{final}}, b_i^{\mathrm{final}}) \gets K_i$
      \State Apply $Z^{b_i^{\mathrm{final}}} X^{a_i^{\mathrm{final}}}$ to qubit $i$ of $qc_{\mathrm{dec}}$
    \EndFor
    \State Measure all qubits of $qc_{\mathrm{dec}}$ (once or multiple times)
    \State Collect post-decryption bitstring $y$ (can repeat per count)
    \State $p_{\mathrm{dec}}(y)\ \gets$  Accumulate count into $y$ in decrypted histogram weighted by the number of times the encrypted result appeared
  \EndFor
  \State \textbf{return} decrypted measurement statistics $\{p_{\mathrm{dec}}(y)\}$, $K$
\EndProcedure
\end{algorithmic}
\end{breakablealgorithm}

The same optional computational-basis input-masking optimization described for Algorithm~\ref{alg:qhe} applies to this workflow.

The QOTP decryption operation is then applied to this state by acting with the corresponding $Z^{b_i} X^{a_i}$ operators on each qubit, using the final QOTP key frame obtained after homomorphic key tracking. The circuit is measured, yielding the decrypted output. By repeating this process for all distinct encrypted bitstrings (and weighting each decrypted outcome by its observed frequency), the final decrypted probability distribution is reconstructed, matching what would have been obtained had the decryption been performed on the quantum device itself.

This makes it possible to decrypt measured classical outcomes locally, keeping the plaintext output labels unavailable to the quantum backend before client-side decryption. In Section~\ref{sec:Results}, we decrypt the results using both a local simulator and an external quantum computer, which serves as an example of a local quantum processor.

\subsection{Generate keys using QTRNG}

In Algorithm~\ref{alg:qhe} and Algorithm~\ref{alg:qhe-local-circuit-decrypt} the keys that we are going to use are generated; the generic way to do this is included in point 2 (QOTP Encryption) of the pseudocode.

This random number generator can be replaced by a quantum random number generator (QTRNG)~\cite{hernandez2024true}, which can provide quantum-origin entropy for generating the QOTP key bits. In practice, standard randomness validation and extraction procedures may be required depending on the physical source.

In an ideal implementation, repeated measurement of qubits prepared in the Hadamard basis produces unbiased random bits. In practice, QRNG outputs may require calibration, statistical validation, and randomness extraction to mitigate device bias. The resulting bits can then be used by the client as QOTP key material, provided they are kept secret.

\subsection{Introduction of Swap Gates}

In this model, the introduction of swap gates before measurement is fully supported. These gates exchange the position of the qubits within the circuit.

As a result, the association between measured bitstrings and logical output labels is randomized by an additional permutation layer. This may make direct interpretation less immediate, but it should be understood as heuristic output-label randomization rather than as a formal cryptographic security guarantee.

First, we randomly introduce swap gates before the measurements (so as not to affect the operations previously applied to the circuit). Second, and finally, we must undo the transformation in the decryption circuit just before the decryption gates. This way the correct clear result is obtained. The graphical representation of the created quantum circuit is shown in Section~\ref{sec:Results}.

When a final random qubit permutation $\pi$ is used, the client first
decrypts the measured labels with the final physical QOTP frame and
then restores the logical qubit order:
\begin{equation}
y
=
\pi^{-1}
\left(
x\oplus\mathbf{j}^{\mathrm{final}}
\right).
\end{equation}

\subsection{
Diagrams of the Proposed Algorithms}

This section shows the workflows of the algorithms described previously in Figures~\ref{fig:w1},\ref{fig:w2}. As can be seen, the overall flow varies considerably between the two algorithms.

Algorithm~\ref{alg:qhe} illustrates the in-circuit decryption workflow used for correctness validation. As we will see later in the Results Section~\ref{sec:Results}, the operations are preserved against encryption. 
In contrast, in Algorithm~\ref{alg:qhe-local-circuit-decrypt}, the measured outcomes returned by the server remain encrypted and are decrypted locally by the client using the final QOTP key frame. As shown in the flowchart, if an adversary only observes the returned measurement outcomes and does not know the final QOTP key frame, the outcomes remain encrypted under the local-decryption workflow.

\onecolumn

\begin{figure}[H]
\centering
\fbox{\includegraphics[width=0.99\linewidth]{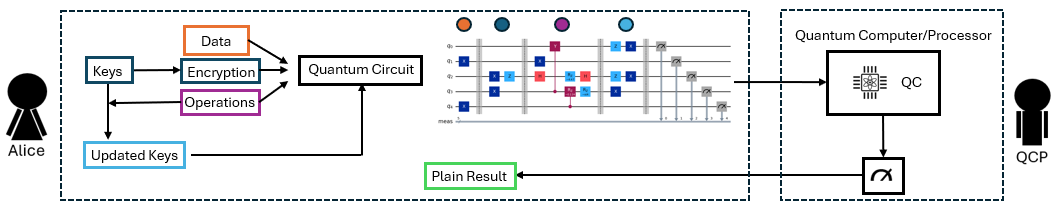}}
\caption{
Algorithm~\ref{alg:qhe} workflow.}

\label{fig:w1}
\end{figure}

\begin{figure}[H]
\centering
\fbox{\includegraphics[width=0.99\linewidth]{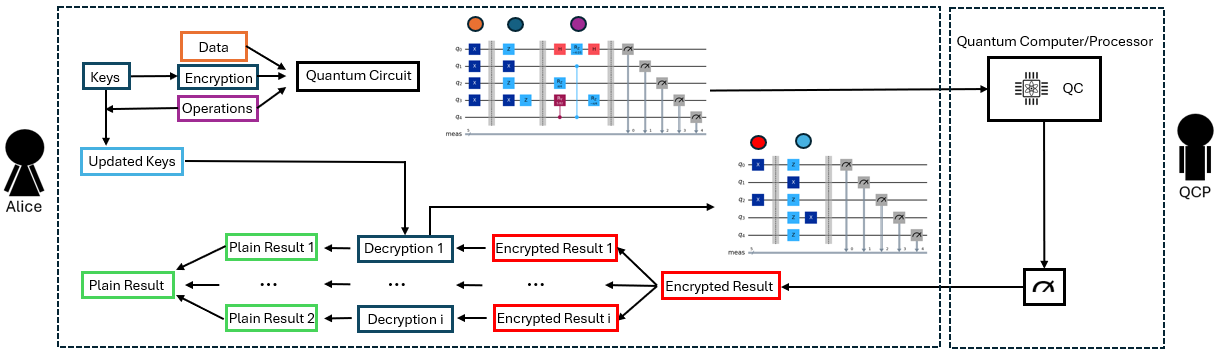}}
\caption{
Algorithm~\ref{alg:qhe-local-circuit-decrypt} workflow.}

\label{fig:w2}
\end{figure}

\twocolumn
Figures~\ref{fig:w1} and~\ref{fig:w2} include example quantum circuits to help understand how the complete Algorithms~\ref{alg:qhe} and~\ref{alg:qhe-local-circuit-decrypt} work, detailing each step, how the circuits are built, and where they are applied. The representation itself includes colors to facilitate the understanding of the different elements of the quantum circuit.
\vspace{0.3cm}

\section{Experimental Validation of Homomorphic QOTP}
\label{sec:Results}

This section presents the experimental validation of the proposed QOTP-based key-tracked compilation framework. All results are derived from executions on both real IBM Quantum devices and noiseless quantum simulators. Specifically, we employed IBM Basquecountry, a 156-qubit superconducting quantum processing unit (QPU) featuring the Heron r2 processor, thereby providing a comprehensive comparison between idealized and hardware-constrained environments. The primary focus is to assess correctness, output-distribution fidelity, and hardware behavior of the compiled encrypted circuits under realistic quantum noise and to quantify the practical gap between simulated and physical quantum computation.
In quantum computing, fidelity-type metrics are commonly used to quantify how similar two output distributions are. Since the experimental outputs considered here are classical histograms obtained from computational-basis measurements, we compare the ideal and homomorphically evaluated output distributions using the Hellinger fidelity. For two probability distributions $P=\{p_i\}$ and $Q=\{q_i\}$, we define
\begin{equation}
F_H(P,Q)
=
\left(
\sum_i \sqrt{p_i q_i}
\right)^2.
\label{eq:hellinger-fidelity}
\end{equation}
This quantity satisfies $0\leq F_H(P,Q)\leq 1$, with $F_H(P,Q)=1$ corresponding to identical distributions. Therefore, values close to one indicate strong agreement between the original and homomorphically evaluated circuits.

Our evaluation is divided into two core experiments: (i) Algorithm~\ref{alg:qhe}, where quantum one-time pad (QOTP) decryption is performed in-circuit on the quantum device prior to measurement, and (ii) Algorithm~\ref{alg:qhe-local-circuit-decrypt}, where quantum circuits return encrypted outputs and all decryption is subsequently performed by quantum circuits, after measurement. For both algorithms, we consider two circuits, one the original and the other its homomorphic version. We calculate the average Hellinger fidelity between the output distributions of the original and homomorphic circuits, allowing for an objective evaluation of correctness, fidelity, and hardware scaling behavior in the tested regimes.

All circuits used in the experiments were generated randomly from the key-tracked gate set summarized in Table~\ref{homomorphic-decompositions}, with gates applied in random order to randomly chosen qubits. This design samples a range of circuit instances while keeping circuit depth moderate and within current hardware capabilities. Due to the external queueing and execution time variability of IBM Quantum hardware, total execution times are not directly comparable or reproducible.

For each hardware configuration reported in the tables, approximately 100 independently generated random circuit instances were evaluated. For each simulated configuration, more than 50,000 independently generated random circuit instances were evaluated. Each circuit instance was executed using the number of shots specified in the corresponding table, and the reported values are the empirical mean Hellinger fidelities over the resulting circuit instances.

\subsection{QOTPH with In-Circuit Decryption}

Table~\ref{tab:algorithm1-fidelity} shows the average Hellinger fidelity obtained from applying the homomorphic evaluation protocol (Algorithm~\ref{alg:qhe}) with decryption performed directly on quantum hardware, immediately before measurement. Experiments were executed on IBM Quantum devices as well as noiseless Qiskit simulators. Fidelities correspond to the overlap between the output distributions of the original circuit and the homomorphic one.

\begin{table}[htbp]
\centering
\caption{\bf Average Hellinger Fidelity of Experimental results for Algorithm~\ref{alg:qhe}}
\begin{tabular}{c c c c c}
\hline
\textbf{Qubits} & \textbf{Gates} & \textbf{Shots} & \textbf{Simulator} & \textbf{QC} \\
\hline
5 & 5 & 10000 & 0.9999 & 0.9912 \\  
5 & 10 & 15000 & 0.9999 & 0.9866 \\ 
5 & 15 & 20000 & 0.9999 & 0.9792 \\ 
10 & 5 & 30000 & 0.9999 & 0.9904 \\ 
10 & 10 & 30000 & 0.9999 & 0.9798 \\ 
10 & 15 & 30000 & 0.9999 & 0.9691 \\ 
20 & 5 & 50000 & - & 0.9938 \\ 
20 & 10 & 50000 & - & 0.9739 \\ 
20 & 15 & 50000 & - & 0.9728 \\ 
40 & 5 & 50000 & - & 0.9726 \\ 
40 & 10 & 50000 & - & 0.9613 \\ 
40 & 15 & 50000 & - & 0.9474 \\ 
\hline
\end{tabular}
\label{tab:algorithm1-fidelity}

\end{table}

The results demonstrate that the protocol preserves correct logical computation and enables successful homomorphic encryption on real quantum hardware, as evidenced by the high fidelity values observed for 5, 10, 20 and 40-qubit circuits. As expected, simulators consistently achieve near-perfect fidelity ($>0.999$), indicating that the limiting factor in hardware experiments is not the protocol itself, but rather intrinsic device noise (decoherence, gate errors, and readout imperfections). Despite these hardware limitations, all measured fidelities remain above $0.94$ for the tested configurations, consistent with the noise levels typically observed in current quantum devices. These results support the practical implementability of the scheme for moderate circuit widths on current superconducting quantum processors.

The following Figure~\ref{fig:fid1} shows the graphical representation of the fidelity results obtained as a function of the qubits and the number of quantum gates.

\begin{figure}[h!]
\centering
\fbox{\includegraphics[width=0.95\linewidth]{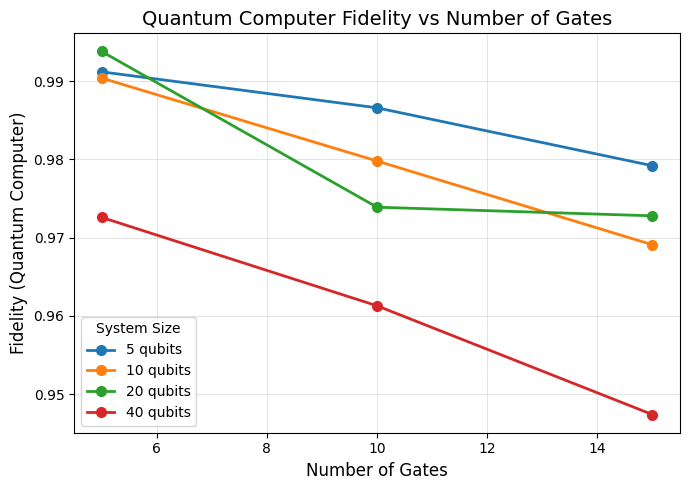}}
\caption{
Fidelity of the results obtained for Algorithm~\ref{alg:qhe}.}

\label{fig:fid1}
\end{figure}

As can be seen, the results show how fidelity decreases depending on the number of qubits and the number of initial gates applied (it should be noted that when the circuit is sent to the quantum computer, the number of gates increases during transpilation depending on the quantum gates supported by the computer). The decreases in fidelity are proportional to what is expected based on the accumulated errors of the quantum computer, remaining above 0.94.

\subsection{Image Based Comparison: Original and Homomorphic Circuits of Algorithm~\ref{alg:qhe}}

Figure~\ref{fig:circuit-comparison} presents a representative example for comparing the structure of an original quantum circuit and its corresponding homomorphic version, generated according to the protocol implementation described in Section~\ref{sec:Implementation}. The top panel depicts a standard circuit, while the bottom panel shows its QOTP-encrypted homomorphic equivalent, including key-dependent gate modifications and QOTP encryption/decryption layers.

\begin{figure}[h!]
\centering
\fbox{\includegraphics[width=0.90\linewidth]{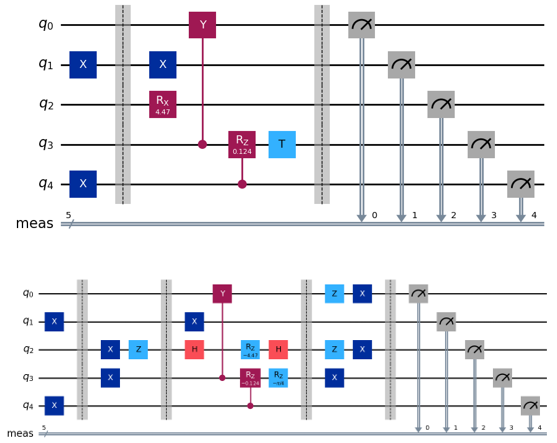}}
\caption{
Visual comparison of a standard quantum circuit and its QOTP-based homomorphic version.}

\label{fig:circuit-comparison}
\end{figure}

When evaluating the quantum circuit on external hardware, the circuit to be sent to it is the one we can see in the lower part of Figure~\ref{fig:circuit-comparison}, since the original is shown as a comparative example. As can be seen, the depth of the homomorphic circuit increases depending on the gates present in the original circuit. This will be discussed in the following sections.

\subsection{QOTPH with Encrypted Output and Local Decryption}

Table~\ref{tab:algorithm2-fidelity} presents the results of executing Algorithm~\ref{alg:qhe-local-circuit-decrypt} with measurement performed on encrypted quantum states, followed by local decryption using the updated QOTPH keys. In this approach, the quantum computer returns encrypted bitstrings and all decryption is performed locally. After decryption, the output distributions closely agree with those obtained using Algorithm~\ref{alg:qhe}, according to the Hellinger fidelity metric, validating the correctness of the key-tracking and local-decryption workflow.

\begin{table}[h!]
\centering
\caption{\bf Average Hellinger Fidelity of Experimental Results for Algorithm~\ref{alg:qhe-local-circuit-decrypt}}
\begin{tabular}{c c c c c}
\hline
\textbf{Qubits} & \textbf{Gates} & \textbf{Shots} & \textbf{Simulator} & \textbf{QC} \\
\hline
5 & 5 & 20000 & 0.9999 & 0.9887 \\ 
5 & 10 & 20000 & 0.9999 & 0.9907 \\
5 & 15 & 20000 & 0.9999 & 0.9802 \\
10 & 5 & 30000 & 0.9999 & 0.9724 \\ 
10 & 10 & 30000 & 0.9999 & 0.9684 \\ 
10 & 15 & 30000 & 0.9999 & 0.9631 \\
20 & 5 & 50000 & - & 0.9391 \\ 
20 & 10 & 50000 & - & 0.9394 \\ 
20 & 15 & 50000 & - & 0.9335 \\ 
\hline
\end{tabular}
\label{tab:algorithm2-fidelity}

\end{table}

The results in Table~\ref{tab:algorithm2-fidelity} show the average Hellinger fidelity of Algorithm~\ref{alg:qhe-local-circuit-decrypt} after applying the local decryption process. After local decryption, the output distributions match those obtained with in-circuit decryption, validating the correctness of the key-tracking and decryption procedures. These experiments demonstrate that encrypted measurement outcomes can be locally decrypted using the final QOTP key frame.

\begin{figure}[ht!]
\centering
\fbox{\includegraphics[width=0.95\linewidth]{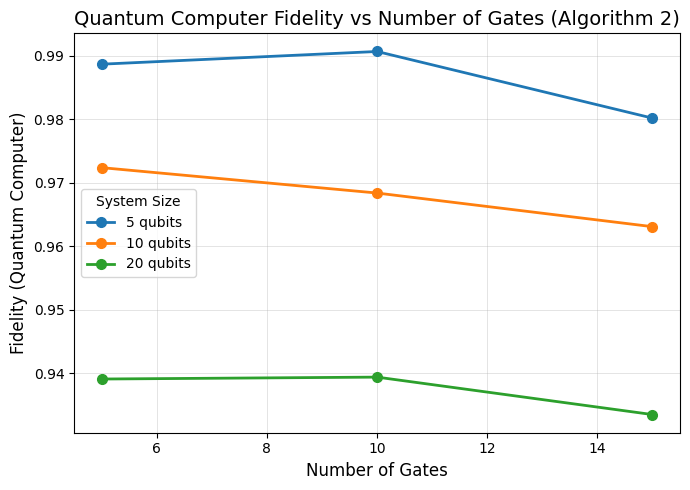}}
\caption{
Fidelity of the results obtained for Algorithm~\ref{alg:qhe-local-circuit-decrypt}.}

\label{fig:fid2}
\end{figure}

In Figures~\ref{fig:fid1} and~\ref{fig:fid2}, markers represent the
experimentally evaluated gate counts ($5$, $10$, and $15$ gates).
Connecting segments are included only as visual guides and do not
represent additional measurements, interpolation, or extrapolation.

Figure~\ref{fig:fid2} shows the fidelity as a function of the number of qubits and initial gates. As in the previous subsection, fidelity decreases as circuit size increases, while remaining above 0.93 for 20 qubits despite the execution of two circuits.

\subsection{Image Based Comparison: Original and Homomorphic Circuits of Algorithm~\ref{alg:qhe-local-circuit-decrypt}}

Figure \ref{fig:circuit-comparison2} presents three quantum circuits. The first one corresponds to the original circuit used as a reference. The second shows its homomorphic counterpart, where the evaluation is performed without embedding the decryption stage. Finally, the third circuit illustrates the decryption procedure applied to one of the states produced after the measurements performed on the external quantum computer.

\begin{figure}[h!]
\centering
\fbox{\includegraphics[width=0.86\linewidth]{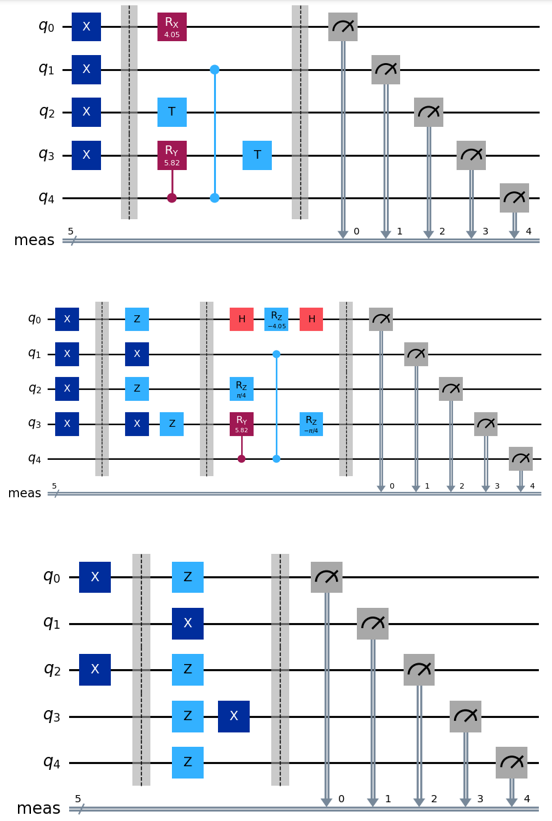}}
\caption{
Visual comparison of a standard quantum circuit, its QOTP-based homomorphic version, and the local decryption circuit of one of the resulting states of Algorithm~\ref{alg:qhe-local-circuit-decrypt}.}

\label{fig:circuit-comparison2}
\end{figure}

When evaluating the quantum circuit on external hardware, the circuit to be sent to it is the one we can see in the middle part of Figure~\ref{fig:circuit-comparison2}, since the original is shown as a comparative example.

\subsection{Image Based Comparison: Original and Homomorphic Circuits of Algorithm~\ref{alg:qhe-local-circuit-decrypt} With Random Swap Gates}

The quantum circuit of Algorithm~\ref{alg:qhe-local-circuit-decrypt} is shown graphically in Figure~\ref{fig:circuit-comparison2Swap} after the introduction of random swap gates as a heuristic output-label randomization technique.

\begin{figure}[h!]
\centering
\fbox{\includegraphics[width=0.89\linewidth]{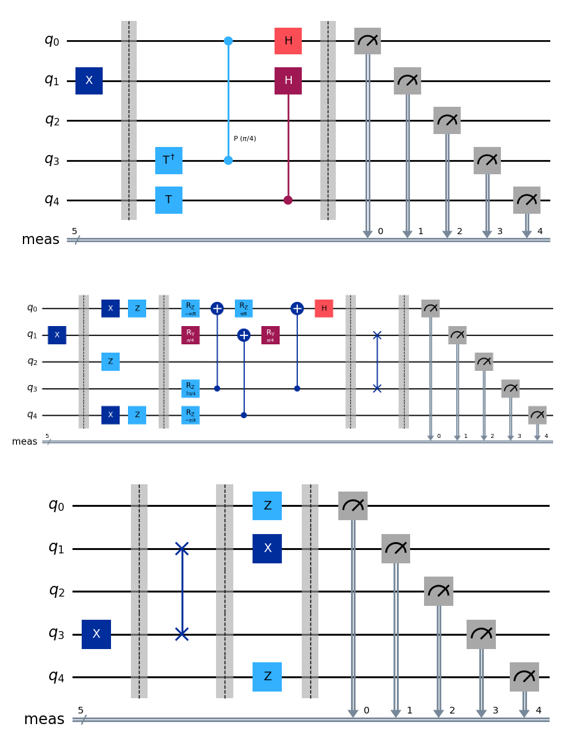}}
\caption{
Visual comparison of a standard quantum circuit, its QOTP-based homomorphic version, and the local decryption circuit of one of the resulting states of Algorithm~\ref{alg:qhe-local-circuit-decrypt} with swap gate.}

\label{fig:circuit-comparison2Swap}
\end{figure}

As previously mentioned, the introduction of swap gates provides an auxiliary output-label randomization layer, making the direct interpretation of measured bitstrings less immediate for an observer without the final qubit-label mapping. We do not treat this mechanism as a standalone cryptographic security guarantee, since the submitted circuit may still be visible to the backend. As we can see in Figure~\ref{fig:circuit-comparison2Swap}, a swap gate has been introduced and is subsequently reversed in the local decryption circuit. In this particular case, two qubits have been swapped.

The logical result is identical to that of the original workflow shown in Table~\ref{tab:algorithm2-fidelity} since the additional permutation does not affect the functioning of the model. In the ideal circuit model, the additional permutation leaves the logical result unchanged once the inverse permutation is applied during decryption. On noisy hardware, however, explicit SWAP operations may introduce additional transpilation and gate overhead and can therefore affect the measured fidelity.

\subsection{Illustrative Analysis of Hardware Fidelity as a Function of Circuit Depth}

Initially, the original quantum circuit consisted of a limited number of gates, while its homomorphic counterpart featured a somewhat greater circuit depth (depending on the original gates chosen).

Notably, simulation results consistently indicate a high fidelity close to 0.999, irrespective of increasing circuit depth. Specifically, fidelity measurements remain stable around 0.999 across all evaluated circuit depths in simulation environments. This outcome demonstrates that in an ideal, noise-free simulator, fidelity is unaffected by circuit complexity. Besides, the fidelity of real quantum computers decreases drastically with circuit depth, as indicated by the devices themselves, which exhibit a certain error depending on the applied gate.

Figure~\ref{fig:depth} provides an auxiliary illustration of the general impact of physical circuit depth on noisy quantum hardware. It compares two logically equivalent circuits with substantially different physical depths.

\begin{figure}[h!]
\centering
\fbox{\includegraphics[width=0.95\linewidth]{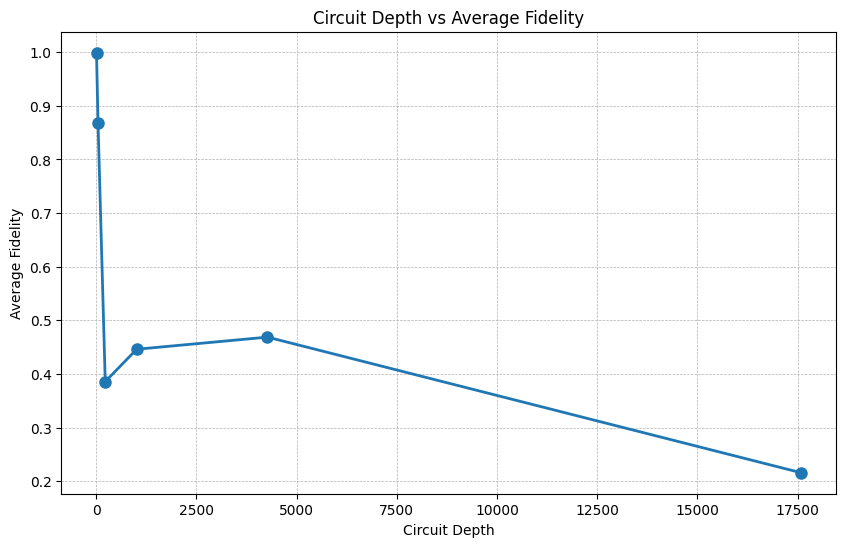}}
\caption{
Hardware output fidelity for two logically equivalent quantum circuits with different physical circuit depths. The deeper implementation contains additional $U3$ and CNOT gates.}

\label{fig:depth}
\end{figure}

This experiment is not intended as a direct benchmark of the QOTPH compiler, but rather as a hardware-level illustration of why deeper compiled implementations may experience greater accumulated noise despite being exactly equivalent in the ideal circuit model.

Consequently, these results underline that the practical scalability of the QOTPH model depends predominantly on the gate-induced noise in real quantum hardware, hence relying on future quantum hardware advances to support deeper compiled circuits with lower accumulated error (a present and well-known problem in literature~\cite{preskill2018quantum}).

\subsection{Experimental Fidelity Analysis as a Function of the Number of Qubits}

To further assess the hardware behavior of the proposed protocol as the number of qubits increases, we analyze the dependence of output fidelity on the number of qubits. Figure~\ref{fig:scaling-fidelity} displays the average Hellinger fidelity achieved for both in-circuit Algorithm~\ref{alg:qhe} and post-processed decryption Algorithm~\ref{alg:qhe-local-circuit-decrypt} as a function of the number of qubits.

\begin{figure}[h!]
\centering
\fbox{\includegraphics[width=0.97\linewidth]{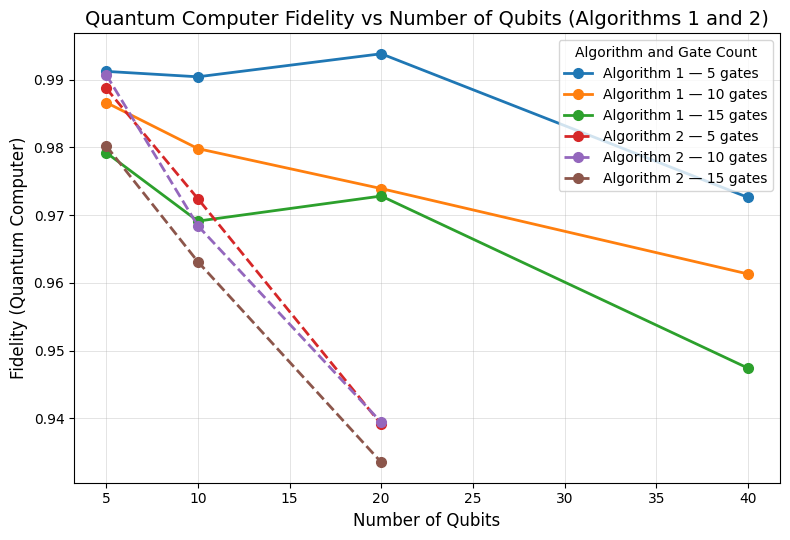}}
\caption{
Average output fidelity versus number of qubits for the quantum homomorphic protocol, for both in-circuit and local (post-processed) decryption.}
\label{fig:scaling-fidelity}
\end{figure}

The results show that, while the fidelity for hardware execution gradually decreases with increasing circuit width, the protocol maintains output fidelity above 0.93 for up to 20 qubits and up to 0.98 for 5 gates on current IBM Quantum hardware. This supports the practical feasibility of the scheme in the tested regimes, while also showing the expected degradation due to realistic device constraints. Simulated results confirm that errors induced by the protocol are negligible and that therefore the proposed QOTPH model works correctly. However, the degradation observed on hardware increases with circuit width and compiled depth and is consistent with accumulated physical noise and implementation overhead. These results support the practical feasibility of the scheme within the experimentally tested regimes.

\subsection{Key Update in QOTPH}

The following Table~\ref{tab:qotph-key-update} provides illustrative examples of how the keys are updated in QOTPH as a result of applying different sets of quantum operations. For each scenario, the table shows the initial keys, the sequence of operations applied, and the resulting updated keys.

\begin{table}[h!]
\centering
\caption{\bf Key Update Examples in QOTPH}
\scriptsize
\renewcommand{\arraystretch}{1.5}
\begin{tabular}{lp{0.7\linewidth}}
\hline
\multicolumn{2}{c}{\textbf{Example 1}} \\
\hline
Key$_{initial}$ & \{0: (1,1), 1: (1,1), 2: (0,0), 3: (0,0), 4: (1,1)\} \\
Operations   & (t,2), (cx,0,2), (cy,4,3), (ccz,1,0,4), (swap,3,2) \\
Key$_{updated}$ & \{0: (1,1), 1: (1,1), 2: (1,1), 3: (1,0), 4: (1,1)\} \\
\hline
\multicolumn{2}{c}{\textbf{Example 2}} \\
\hline
Key$_{initial}$ & \{0: (0,1), 1: (0,0), 2: (1,0), 3: (0,1), 4: (1,1)\} \\
Operations   & (s,4), (ccx,4,2,1), (cy,0,3), (crz,4,2,0.25), (swap,1,3) \\
Key$_{updated}$ & \{0: (0,0), 1: (0,1), 2: (1,0), 3: (0,0), 4: (1,0)\} \\
\hline
\multicolumn{2}{c}{\textbf{Example 3}} \\
\hline
Key$_{initial}$ & \{0: (0,1), 1: (1,0), 2: (0,0), 3: (0,0), 4: (1,0)\} \\
Operations   & (rzz,0,2,3.70), (rz,3,0.72), (cxx,1,3,4), (sqrtx,1), (tdg,0) \\
Key$_{updated}$ & \{0: (0,1), 1: (1,0), 2: (0,0), 3: (1,0), 4: (0,0)\} \\
\hline
\end{tabular}
\label{tab:qotph-key-update}
\end{table}

\newpage
\subsection{Comparative Analysis and Discussion}

The comparative results of Algorithms~\ref{alg:qhe} and~\ref{alg:qhe-local-circuit-decrypt} lead to the following observations:

\begin{itemize}
    \item \textbf{Correctness:} In all tested cases, the decrypted output distributions closely reproduce the corresponding reference distributions, both when decryption is performed in-circuit and when applied locally after encrypted measurement.

    \item \textbf{Hellinger fidelity and hardware behavior:} The average Hellinger fidelity remains high in noiseless simulation and decreases on real hardware as the number of qubits and gates increases. This behavior is consistent with accumulated hardware noise, two-qubit gate errors, and transpilation overhead. 

    \item \textbf{Local decryption:} The agreement between in-circuit and local-decryption variants confirms that the final QOTP key frame is tracked consistently and can be used to recover measured classical outcomes after encrypted execution.

    \item \textbf{Security-related interpretation:} The delegated workflow keeps the QOTP frame, logical circuit, and logical-to-physical mapping client-side. An isolated compiled operation does not directly determine a unique originating logical gate or QOTP-frame value without the corresponding signed logical information. Section~\ref{subsec:algorithm-aware-analysis} further characterizes the additional conditions under which algorithmic or instance-specific side information can reduce this ambiguity.

    \item \textbf{Practicality:} The results show that QOTP key-tracked compilation can be implemented in a standard quantum software stack and executed on current IBM Quantum hardware, although current performance remains limited by noise and circuit depth.
\end{itemize}

Overall, the experiments support the practical feasibility of the proposed QOTP key-tracked compilation framework as an implementation and correctness-validation tool. The protocol delivers high-fidelity correctness validation for the tested circuit sizes and helps quantify the practical gap between ideal QOTP key tracking and execution on current noisy quantum hardware.

\section{Security Analysis}
\label{sec:security}

The proposed QOTPH framework inherits the information-theoretic confidentiality of the Quantum One-Time Pad at the level of quantum ciphertexts. For an $n$-qubit state $\rho$ encrypted with uniformly random QOTP keys, the averaged ciphertext satisfies
\begin{equation}
\frac{1}{4^n}
\sum_{j,k\in\{0,1\}^n}
X^j Z^k \rho Z^k X^j
=
\frac{I}{2^n}.
\label{eq:qotp-twirl-security}
\end{equation}
Therefore, any party that does not know the QOTP key obtains no information about the plaintext quantum state from the encrypted quantum register alone. The client tracks the QOTP key frame gate by gate, so the evaluated register remains interpretable as a QOTP-encrypted version of the corresponding logical state. This provides information-theoretic hiding of the quantum data itself, assuming that the QOTP keys remain uniformly random and secret.

The evaluator's view consists of the encrypted quantum register and the flattened compiled circuit. In the delegated model considered here, the server is not given the uncompiled logical circuit, the QOTP frame, or the correspondence between logical and physical instructions. The compiled gate list is therefore generally compatible with multiple
logical-circuit and key-frame configurations. For example, an observed $R_z(\alpha)$ gate is compatible with either $(\theta=\alpha,j=0)$ or $(\theta=-\alpha,j=1)$.

This ambiguity means that an isolated corrected gate is not a direct encoding of a unique QOTP key bit. It does not, however, establish general key-hiding against arbitrary side information. If the exact logical circuit and the relevant signed logical parameters can be associated with the compiled instructions, corrected signs or branch structures may reveal key-frame bits or parities. The security statement of QOTPH therefore combines information-theoretic hiding of the quantum-register marginal with compiled-circuit key ambiguity under the hidden-logical-circuit model.

For computational-basis inputs, the initial $X$-mask can be incorporated into the locally prepared bitstring, avoiding a separate visible mask layer. In the local-decryption workflow, the server returns outcomes
whose labels remain masked by the final $X$-frame, and the client restores the logical labels locally.

QOTPH also supports an alternative statevector-based preparation mechanism in which the client locally constructs and encrypts the input state, extracts the resulting encrypted statevector, discards the temporary preparation circuit, and submits only a fresh circuit initialized from that encrypted state. This removes the original input-encoding and explicit QOTP-encryption gates from the evaluator's gate-level view and prevents direct recovery of the QOTP bits by simply inspecting a separate visible mask layer. However, the encrypted statevector or the state-preparation circuit synthesized from it may remain inspectable by the evaluator. This mechanism therefore hides the preparation sequence and the explicit encryption-gate layer rather than establishing general cryptographic key hiding.

When confidentiality of the submitted circuit description is additionally required, a separate circuit-hiding mechanism is necessary. Blind quantum computation constitutes an established cryptographic direction for addressing this additional requirement. Universal Blind Quantum Computation (UBQC) is a canonical example of an information-theoretically blind delegation protocol~\cite{broadbent2009universal}. Its execution model, however, differs substantially from QOTPH: UBQC relies on an interactive measurement-based protocol and additional client-side quantum-state preparation and communication requirements, whereas QOTPH targets client-side QOTP key tracking and compilation followed by direct non-adaptive execution on conventional gate-based quantum hardware. We therefore cite blind quantum computation as an established approach to the separate problem of computation hiding, rather than as a component or security property of the QOTPH protocol analyzed and experimentally validated in this work.

The low pre-decryption fidelity values observed in many instances indicate that the encrypted outcome labels differ substantially from the corresponding plaintext output labels before local decryption.
Higher values may occur when the QOTP-induced relabeling preserves some features of the distribution, for example when the output is highly concentrated or when different labels have similar probabilities. These fidelity values are not a cryptographic security parameter and do not establish an indistinguishability guarantee; rather, they provide an empirical illustration of the distinction between encrypted measurement outcomes and the logical outputs recovered through local decryption.

It is also important to note that QOTP-based encryption provides confidentiality but not integrity. A malicious evaluator could deviate from the prescribed computation by omitting gates, inserting extra operations, or performing unintended measurements. The experimental results in this work validate correctness and quantify the impact of hardware noise, but they do not by themselves provide verifiability against active adversaries.

Despite these limitations, QOTPH offers several security-relevant contributions in practice: (i) information-theoretic hiding of the quantum ciphertexts, assuming secret uniformly random keys; (ii) client-side resolution of all key-dependent gate corrections, so that the server does not receive or update the QOTP key frame; and (iii) local decryption of encrypted measurement outcomes using the final key frame. 
Accordingly, the direct QOTPH execution mode provides QOTP-based data confidentiality and client-side key-tracked compilation, while the submitted classical circuit description is not claimed to satisfy general circuit blindness. Established blind-computation protocols illustrate how this additional privacy objective can be addressed under different execution and client-resource assumptions. Integrity and verifiability against actively deviating evaluators remain separate requirements.

\subsection{Algorithm-aware key-inference analysis}
\label{subsec:algorithm-aware-analysis}

To assess whether knowledge of the algorithmic family can facilitate inference of the QOTP key, we distinguish knowledge of an algorithm from knowledge of its exact logical instance. Knowing that a circuit implements a particular quantum algorithm does not in general reveal the signed values of its logical rotation parameters, oracle-dependent operations, problem-instance data, variational parameters, or exact gate decomposition. This distinction is essential for QOTPH because the observable parameter of a compiled rotation combines the logical parameter with a QOTP-dependent sign.

For example, for a logical $R_z(\theta)$ operation, QOTPH submits a physical rotation $R_z(\alpha)$ satisfying
\begin{equation}
\alpha=(-1)^j\theta.
\end{equation}
Observation of $\alpha$ therefore reveals $j$ only when the signed logical value $\theta$ is independently known. If the evaluator knows only the algorithmic family, or knows $|\theta|$ but not its logical sign, the two explanations $(j=0,\theta=\alpha)$ and $(j=1,\theta=-\alpha)$ remain compatible with the same observed physical rotation. Analogous ambiguities apply to $R_x$, $R_y$, and Pauli-generated two-qubit rotations.

To account for unknown logical parameter signs, let $s_r\in\mathbb{F}_2$ denote the unknown sign associated with the $r$-th independent logical parameter class. Each observable key-dependent rotation then induces a binary relation involving both QOTP-frame variables and, when required, an unknown logical-sign variable. For example,
\begin{align}
R_z &: \quad b = j_q \oplus s_r,\\
R_x &: \quad b = k_q \oplus s_r,\\
R_y &: \quad b = j_q \oplus k_q \oplus s_r,\\
R_{ZZ} &: \quad b = j_p \oplus j_q \oplus s_r,
\end{align}
where $b$ denotes the observed sign of the corresponding compiled parameter relative to a chosen positive-magnitude convention.

These sign relations apply when the submitted parameter convention distinguishes $+\theta$ from $-\theta$; parameter values for which both signs are equivalent modulo gate periodicity or global phase do not provide an independent sign observation.

The resulting inference problem depends on the type of logical parameter involved. We distinguish three relevant cases. First, for free or instance-dependent parameters whose signed logical values are unknown to the evaluator, an isolated corrected rotation retains the ambiguity between the logical sign and the QOTP-dependent sign. Second, when the same unknown logical parameter is reused across several operations, comparison of the corresponding compiled signs may eliminate the common unknown sign and reveal relative QOTP-frame parities. Third, fixed-angle operations whose logical sign is determined by the recognized algorithmic structure, such as $T$ gates or fixed controlled-phase rotations, do not benefit from this parameter-sign ambiguity and may therefore expose QOTP-frame information if their correspondence with the compiled operations can be identified.

Because all QOTPH key-frame updates are linear transformations over $\mathbb{F}_2$, every current frame bit appearing in these relations can be expressed as a linear function of the initial frame $K_0$.

Collecting all observable relations yields a binary linear system
\begin{equation}
A_K K_0 \oplus A_S S = b,
\end{equation}
where $K_0\in\mathbb{F}_2^{2n}$ is the initial QOTP key and $S$ contains the unknown logical-parameter signs. The number of independent constraints obtained on the QOTP key after eliminating the nuisance sign variables is
\begin{equation}
r_K =
\operatorname{rank}
\begin{pmatrix}
A_K & A_S
\end{pmatrix}
-
\operatorname{rank}(A_S).
\label{eq:key-leakage-rank}
\end{equation}

Thus, an isolated rotation with an otherwise unknown logical sign contributes no independent key information, whereas repeated occurrences sharing the same unknown logical parameter may reveal relative QOTP-frame parities because the common sign variable cancels.

Parametrized controlled rotations require separate consideration because the QOTPH compilation rule may modify not only the sign of a rotation parameter but also the structure of the submitted gate sequence. 

If a compiled subsequence can be reliably associated with a particular logical controlled rotation, the distinction between the $j_c=0$ and $j_c=1$ compilation branches can reveal the current control $X$-frame bit independently of the unknown sign of the logical rotation parameter. Consequently, controlled rotations constitute a stronger potential leakage channel than isolated single-qubit sign corrections.

Because the number of independent key constraints depends on the specific problem instance, logical decomposition, parameter reuse, and the evaluator's ability to associate physical subsequences with logical operations, the following analysis is deliberately structural rather than quantitative. We consider only information obtainable from the submitted classical circuit description and do not use measurement outcomes as an inference channel. Table~\ref{tab:algorithm-key-inference} therefore identifies the principal theoretical key-inference mechanisms for representative algorithms when the algorithmic family is known but the complete signed logical circuit is not.

In the direct QOTPH execution model, the compiled classical circuit description is submitted directly to the evaluator. We additionally analyze the information that may be inferred when this circuit is available in clear form. In this setting, the evaluator may inspect the physical gate types, rotation parameters, gate ordering, and qubit connectivity of the submitted QOTPH circuit, and is additionally assumed to know the family of quantum algorithm being executed. No blind-computation, circuit-hiding, or additional obfuscation layer is assumed in this analysis. However, knowledge of the algorithmic family does not imply knowledge of the exact logical instance, including its signed logical parameters, problem-dependent oracles, variational parameters, or the correspondence between logical and compiled instructions. We consider only information obtainable from the submitted classical circuit description and do not use measurement outcomes as an inference channel.

Overall, Table~\ref{tab:algorithm-key-inference} shows that visibility of the compiled QOTPH circuit and knowledge of the algorithmic family do not, by themselves, generally determine the QOTP key. For Clifford-dominated constructions, such as the standard Bernstein--Vazirani circuit, the QOTPH rules introduce no direct key-dependent angle or conditional-branch information. For algorithms containing instance-dependent or variational rotations, such as VQE, QAOA, and Hamiltonian-simulation circuits, unknown signed logical parameters preserve ambiguity between logical-parameter signs and QOTP-dependent corrections, although parameter reuse may reveal relative frame parities. Stronger inference becomes possible for identifiable fixed-angle structures, such as QFT-based blocks, or for controlled rotations whose QOTPH compilation changes structurally with the current QOTP frame. Therefore, even when the submitted compiled circuit is directly inspectable and the algorithmic family is known, general key recovery requires additional information concerning the logical instance, signed parameter values, decomposition, or logical-to-compiled instruction correspondence.

\begin{table*}[!t]
\centering
\caption{\bf Structural analysis of QOTP key inference for representative quantum algorithms in the direct QOTPH execution model. The compiled circuit is assumed to be visible to the evaluator, without an additional circuit-hiding layer.}
\label{tab:algorithm-key-inference}
\scriptsize
\setlength{\tabcolsep}{1.8pt}
\renewcommand{\arraystretch}{0.9}
\begin{tabular}{p{1.5cm}p{3cm}p{4.5cm}p{8.4cm}}
\hline
Algorithm & Relevant standard circuit structure & Information not implied by algorithm-family knowledge & Additional conditions enabling QOTPH key inference \\
\hline
Bernstein--Vazirani & $H$, $X$, and CX gates in the standard oracle construction & The hidden string determines the oracle instance & The standard construction is Clifford, so the QOTPH compilation rules introduce no direct key-dependent rotation-angle or conditional-branch information. The oracle structure may reveal the problem instance, but this does not by itself reveal QOTP bits through the QOTPH correction rules. \\
Deutsch--Jozsa & Hadamard layers and an oracle $U_f$ & The selected constant or balanced function and its reversible decomposition need not be known & The fixed Hadamard layers introduce no direct key-dependent information. QOTP-frame inference requires identifiable non-Clifford, parametrized, or conditionally compiled operations in the particular implementation of the oracle. Knowledge of the Deutsch--Jozsa family alone is therefore insufficient. \\
Simon & Hadamard layers and a reversible oracle $U_f$ & The hidden-period function and its reversible implementation are instance-dependent & The fixed Hadamard layers introduce no direct key-dependent information. Any potential QOTP inference depends on the decomposition of the instance-specific oracle and on the evaluator's ability to associate its compiled subsequences with corresponding logical operations. \\
Grover search & Phase oracle and diffusion operator, commonly involving multi-controlled operations & The marked-state oracle remains instance-dependent & The recognizable diffusion operator may contain fixed-angle non-Clifford or multi-controlled decompositions. These can generate key-dependent sign or structural constraints if their logical-to-compiled correspondence can be identified. The contribution of the search oracle remains dependent on the marked-state instance and its implementation. \\
Quantum Fourier Transform & Hadamard gates, SWAPs, and controlled-phase operations with fixed structured angles & Forward versus inverse QFT and the correspondence between logical operations and compiled subsequences must still be identified & If a QFT or inverse-QFT block and its logical-to-compiled correspondence are identified, its fixed controlled-phase angles remove the unknown-logical-sign ambiguity associated with free parameters. The corresponding QOTPH-supported decompositions may then provide constraints on current QOTP-frame bits or parities. \\
Quantum Phase Estimation & Controlled powers $U^{2^k}$ followed by an inverse QFT & The unitary $U$, its decomposition, and the underlying problem instance may remain unknown & The identifiable inverse-QFT component contains fixed-angle structure and can provide QOTP-frame constraints if its compiled operations can be associated with their logical counterparts. Inference from the controlled-$U$ part additionally requires knowledge of the particular unitary and its decomposition. \\
Shor order finding & Controlled modular exponentiation followed by inverse QFT & Knowledge of Shor's algorithm does not uniquely determine the selected arithmetic implementation, logical decomposition, or necessarily all instance-dependent choices & The inverse-QFT component contains recognizable fixed-angle structure. Additional QOTP-frame constraints may arise from modular-arithmetic blocks if the arithmetic instance, decomposition, and logical-to-compiled correspondence are also independently known or reliably inferred. \\
Quantum counting & Controlled Grover iterations combined with quantum phase estimation & The search oracle remains problem-dependent & The phase-estimation and inverse-QFT components contain recognizable fixed-angle structure that may provide QOTP-frame constraints if identified. Additional inference from the Grover component depends on knowledge of the underlying oracle and its implementation. \\
Canonical quantum amplitude estimation & Controlled powers of an amplitude-amplification operator followed by inverse QFT & State preparation, the marking oracle, and the resulting amplitude-amplification operator may remain problem-dependent & In the canonical QPE-based construction, the inverse-QFT block can provide fixed-angle constraints if its correspondence is identified. Inference from the amplitude-amplification operator additionally depends on knowledge of its problem-specific state-preparation and oracle components. \\
HHL & Phase estimation of a matrix-derived unitary, eigenvalue-dependent controlled ancilla rotations, and inverse phase estimation & Matrix-dependent evolution, encoded spectral information, and the resulting controlled-rotation parameters depend on the linear-system instance & Recognizable fixed components of the phase-estimation procedure may provide structured constraints. Eigenvalue-dependent controlled rotations retain parameter ambiguity unless the corresponding problem data and signed logical parameters are independently known; their conditional QOTPH compilation can provide stronger information if the logical-to-compiled correspondence is identified. \\
VQE & Parameterized ansatz interleaved with entangling operations & Ansatz parameters are determined by a classical optimization process and their signed values are not implied by knowledge that VQE is being executed & An isolated variational rotation retains ambiguity between the unknown logical-parameter sign and the QOTP-dependent sign. QOTP-frame information may nevertheless arise when parameters are reused, when relations between parameters are independently known, or when fixed-angle non-Clifford components of the selected ansatz can be identified. \\
QAOA & Repeated cost/mixer layers with $\gamma_l$ and $\beta_l$ & Optimized signed parameters and possibly problem-Hamiltonian details & Unknown parameter signs prevent direct absolute-bit inference from isolated rotations; shared layer parameters can reveal relative frame parities, particularly when the associated Hamiltonian coefficients are known. \\
Quantum walks & Walk operators determined by the chosen graph, oracle, coin or reflection operators, and walk formulation & The graph, marked vertices, transition structure, and circuit realization may all be instance-dependent & There is no algorithm-family-wide key-inference rule for quantum walks. Potential QOTP leakage depends on whether the particular walk implementation contains identifiable fixed-angle non-Clifford operations, parametrized rotations, or conditional decompositions and whether these can be associated with their logical counterparts. \\
QSP / QSVT & Signal or block-encoded unitaries interleaved with a structured sequence of phase rotations & Phase factors are determined by the selected polynomial transformation and may or may not be independently known to the evaluator & If the phase sequence is known or can be derived from the target polynomial transformation, its signed logical angles can provide QOTP-frame constraints when associated with the compiled rotations. If the phase factors are unavailable, individual corrected rotations retain ambiguity between the logical phase sign and the QOTP-dependent correction. \\
Hamiltonian simulation / Trotterization & Pauli-generated rotations implementing exponentials of Hamiltonian terms & Hamiltonian coefficients, evolution time, and the selected decomposition may not be known & If these quantities are known, the signed logical rotation parameters may become known and the observed QOTPH-corrected angles can constrain the QOTP frame. Otherwise, isolated rotations retain ambiguity between the logical-parameter sign and the QOTP-dependent sign. \\
\end{tabular}
\end{table*}
\clearpage

\section{Conclusion}
\label{sec:Conclu}

In this work, we have developed and experimentally validated a homomorphic QOTP-based, key-tracked client-side compilation framework for encrypted quantum computation (QOTPH). The construction uses the Quantum One-Time Pad as its cryptographic foundation and tracks the evolution of the QOTP key frame under a broad set of Clifford, non-Clifford, controlled, and parametrized operations. The resulting framework allows the client to compile a corrected physical circuit locally, while the quantum server executes the submitted circuit non-adaptively and without access to the QOTP keys.

A central part of the work is the explicit treatment of the non-Clifford obstruction. While Clifford gates preserve the Pauli frame and admit direct key updates, non-Clifford gates such as $T$ require key-dependent corrections. Rather than assuming that the server can compute or choose these corrections, the proposed framework performs them during client-side compilation. This clarifies the operational role of the client and the server and separates QOTP-based data protection from full privacy of the submitted classical circuit description.

The implementation, developed in Qiskit, was evaluated using both noiseless simulators and IBM Quantum hardware. The experiments validate correctness by comparing the output distributions of the original and homomorphically compiled circuits. We also compare in-circuit decryption with local decryption of encrypted measurement outcomes, showing that the QOTP key frame can be consistently tracked and used to recover the intended measured output distribution.

The experimental results demonstrate that QOTP key-tracked compilation is practically implementable using current quantum software stacks and executable on near-term quantum hardware. Across the tested configurations, the results validate the correctness of the compilation rules, the consistency of local key tracking and output decryption, and the feasibility of executing the resulting encrypted computations on real gate-based processors, with the remaining performance degradation primarily associated with device noise, circuit depth, and gate errors.

Beyond correctness and hardware validation, we performed a structural analysis of potential QOTP key-inference channels when the evaluator has access to the compiled circuit and may know the algorithmic family being executed. The analysis shows that such knowledge does not, by itself, generally determine the QOTP key: instance-dependent or variational parameters can preserve ambiguity between logical-parameter signs and QOTP-dependent corrections. At the same time, identifiable fixed-angle structures, reused parameters, and key-dependent controlled-rotation branches can provide constraints on the QOTP frame when sufficient logical or logical-to-compiled correspondence is available. These results characterize more precisely the security boundary of the direct QOTPH execution model and identify the conditions under which additional circuit-hiding mechanisms become relevant.

Overall, QOTPH provides a unified formal and experimental framework that connects QOTP key propagation, explicit gate-level compilation, local output decryption, and execution on current gate-based hardware.
The compiler is directly applicable under the hidden-logical-circuit model considered here and can serve as a practical base layer for future integrations with FHE-based key handling, blind computation, authentication, or verification mechanisms.

In summary, QOTPH bridges algebraic QOTP key-tracking techniques and their practical realization by providing a unified formal, compilation, and experimental framework for encrypted quantum computation on current gate-based hardware. The framework also provides a natural foundation for future integration with circuit-hiding, verification, authentication, and cryptographic key-handling mechanisms.

\section{Further Work}

While this work demonstrates the practical feasibility of QOTP-based key-tracked client-side compilation for encrypted quantum computation, several important directions remain open.

First, a central direction is the integration or adaptation of circuit-hiding mechanisms for the compiled circuits submitted to the quantum server. In the present framework, key-dependent corrections are computed locally by the client, but the resulting classical gate list may still contain structural information, corrected rotation angles, or topology-dependent features. Possible directions include compatibility with established blind quantum computation protocols, randomized compilation techniques, and hybrid approaches combining QOTP key tracking with classical cryptographic tools. Such extensions would require a separate analysis of their interaction, client-resource, communication, and hardware overheads.

Second, the current protocol is analyzed under an honest-but-curious execution model. Extending the framework to malicious or actively deviating quantum servers requires integrity and verifiability mechanisms. A malicious evaluator could omit gates, insert additional operations, alter measurement procedures, or return falsified results. Future versions of the protocol should therefore investigate compatibility with quantum authentication, verifiable delegated quantum computation, trap-based verification, or other certification techniques.

Third, the experimental implementation should be extended with a more detailed resource and noise analysis. Relevant quantities include circuit depth before and after compilation, number of key-dependent rotations, $T$-count, $T$-depth, transpiled depth, and the effect of backend topology. Future benchmarks should also include comparisons against unencrypted circuits of similar depth in order to separate the overhead of QOTP key tracking from ordinary hardware noise.

Finally, as quantum hardware advances toward lower error rates and eventually fault-tolerant operation, QOTP-based key-tracked compilation may serve as a useful experimental baseline for studying encrypted quantum computation on remote devices. In this setting, the present work provides an implementation-oriented foundation upon which stronger circuit-hiding, verification, and authentication layers can be built.

\paragraph{Acknowledgments} This work was supported by TECNALIA Research \& Innovation through the predoctoral grant No.\ 112794 and by the BIKAINTEK program through scholarship No.\ 017-B2/2024.

We acknowledge the use of IBM Quantum services for this work.

The views expressed are those of the authors and do not reflect the official policy or position of IBM or the IBM Quantum team.

\paragraph{Disclosures}  A patent application related to this work has been filed and is currently pending.

The authors declare no conflicts of interest. 

\paragraph{Data Availability Statement} The numerical data supporting the results presented in this work, including processed fidelity values and representative circuit instances, are available from the corresponding author upon reasonable request.

\paragraph{Code Availability Statement} The implementation used to generate, compile, and evaluate the QOTPH circuits is available from the corresponding author upon reasonable request.

\appendix

\section{Verification of the Primitive Compilation Identities}
\label{app:correctness-proofs}

This appendix verifies the local identity
\begin{equation}
\widetilde{G}P_K
=
e^{i\phi}P_{K'}G
\label{eq:appendix-local-identity}
\end{equation}
for the primitive compilation rules covered by Theorem~\ref{thm:qotph-correctness}. Composite gates are correct whenever their stated decompositions are exact and every primitive gate in the decomposition satisfies Eq.~\eqref{eq:appendix-local-identity}.

\subsection{Clifford gates}

Let $\mathcal{P}_n$ denote the $n$-qubit Pauli group. Every Clifford operator $C$ normalizes $\mathcal{P}_n$. Therefore, for every QOTP frame $K$, there exist an updated frame $K'$ and a
phase $\omega(C,K)$, with $|\omega(C,K)|=1$, such that
\begin{equation}
CP_KC^{\dagger}
=
\omega(C,K)P_{K'}.
\label{eq:appendix-clifford-conjugation}
\end{equation}
Multiplying Eq.~\eqref{eq:appendix-clifford-conjugation} from the
right by $C$ gives
\begin{equation}
CP_K
=
\omega(C,K)P_{K'}C,
\label{eq:appendix-clifford-correctness}
\end{equation}
which is the required local identity.

For the single-qubit Clifford gates used by QOTPH, the corresponding key-frame transformations are
\begin{align}
H:
\quad
(j,k)
&\longmapsto
(k,j),
\\
S:
\quad
(j,k)
&\longmapsto
(j,k\oplus j),
\\
\sqrt{X}:
\quad
(j,k)
&\longmapsto
(j\oplus k,k).
\end{align}
The logical Pauli gates $X$, $Y$, and $Z$ leave the frame unchanged up to an irrelevant phase. The CNOT transformation is derived explicitly in Section~\ref{subsec:cnot-key-propagation-example}. The CZ and SWAP rules follow analogously from their Pauli conjugation relations, while the CY frame update follows from the equivalent identity $\mathrm{CY}_{c,t}=(I\otimes S)\mathrm{CX}_{c,t}(I\otimes S^\dagger)$.

\subsection{Pauli-generated rotations}

Let $Q$ be a Hermitian Pauli operator satisfying $Q^2=I$, and define
\begin{equation}
R_Q(\theta)
=
e^{-i\theta Q/2}.
\label{eq:appendix-pauli-rotation}
\end{equation}
For a QOTP frame $K$, define the commutation parity
$\chi_Q(K)\in\{0,1\}$ by
\begin{equation}
QP_K
=
(-1)^{\chi_Q(K)}P_KQ.
\label{eq:appendix-commutation-parity}
\end{equation}
It follows that
\begin{equation}
P_KQP_K^{\dagger}
=
(-1)^{\chi_Q(K)}Q.
\end{equation}
Consequently,
\begin{align}
P_KR_Q(\theta)P_K^{\dagger}
&=
\exp\left(
-\frac{i\theta}{2}
P_KQP_K^{\dagger}
\right)
\nonumber\\
&=
R_Q\left(
(-1)^{\chi_Q(K)}\theta
\right).
\end{align}
Multiplying from the right by $P_K$ gives
\begin{equation}
R_Q\left(
(-1)^{\chi_Q(K)}\theta
\right)P_K
=
P_KR_Q(\theta).
\label{eq:appendix-pauli-rotation-correctness}
\end{equation}

For the rotations used in QOTPH, the corresponding parities are
\begin{align}
\chi_X(K_i)
&=
k_i,
\\
\chi_Y(K_i)
&=
j_i\oplus k_i,
\\
\chi_Z(K_i)
&=
j_i,
\\
\chi_{X_pX_q}(K)
&=
k_p\oplus k_q,
\\
\chi_{Y_pY_q}(K)
&=
j_p\oplus k_p\oplus j_q\oplus k_q,
\\
\chi_{Z_pZ_q}(K)
&=
j_p\oplus j_q.
\end{align}
Equation~\eqref{eq:appendix-pauli-rotation-correctness} proves the
angle-correction rules for
$R_x$, $R_y$, $R_z$, $R_{XX}$, $R_{YY}$, and $R_{ZZ}$.
The rules for $T$, $T^{\dagger}$, and $S^{\dagger}$ follow from
\begin{equation}
\begin{aligned}
T
&\simeq
R_z(\pi/4),
&
T^{\dagger}
&\simeq
R_z(-\pi/4),
&
S^{\dagger}
&\simeq
R_z(-\pi/2),
\end{aligned}
\end{equation}
where $\simeq$ denotes equality up to a global phase.

\subsection{Controlled Pauli rotations}

Let
\begin{equation}
\begin{aligned}
C_Q(\theta)
&=
|0\rangle\!\langle0|_c\otimes I_t
\\
&\quad +
|1\rangle\!\langle1|_c
\otimes R_Q(\theta)_t
\end{aligned}
\label{eq:appendix-controlled-rotation}
\end{equation}
be a controlled Pauli rotation. Define
\begin{equation}
\theta_K
=
(-1)^{\chi_Q(K_t)}\theta,
\end{equation}
where $K_t$ denotes the target-qubit frame. The compiled operation is
\begin{equation}
\widetilde{C}_Q(\theta,K)
=
\begin{cases}
C_Q(\theta_K),
&
j_c=0,
\\[1mm]
R_Q(\theta_K)_t C_Q(-\theta_K),
&
j_c=1.
\end{cases}
\label{eq:appendix-controlled-rule}
\end{equation}

When $j_c=0$, correctness follows directly from
Eq.~\eqref{eq:appendix-pauli-rotation-correctness}. When $j_c=1$,
direct evaluation in the computational basis of the control gives
\begin{equation}
\begin{aligned}
&
R_Q(\alpha)_t C_Q(-\alpha)
(X_c\otimes I_t)
\\
&\qquad =
(X_c\otimes I_t)C_Q(\alpha).
\end{aligned}
\label{eq:appendix-control-x-identity}
\end{equation}
The control $Z$-mask commutes with $C_Q(\theta)$. Combining
Eqs.~\eqref{eq:appendix-pauli-rotation-correctness} and
\eqref{eq:appendix-control-x-identity} proves the compilation rules
for $CR_x(\theta)$, $CR_y(\theta)$, and $CR_z(\theta)$.

\subsection{Composite gates}

Suppose that a logical gate $G$ has an exact decomposition
\begin{equation}
V_rV_{r-1}\cdots V_1
=
e^{i\alpha_G}G,
\label{eq:appendix-gate-decomposition}
\end{equation}
where every primitive gate $V_s$ satisfies
Eq.~\eqref{eq:appendix-local-identity}. Let $K^{(0)}=K$ and let
\begin{equation}
\widetilde{V}_sP_{K^{(s-1)}}
=
e^{i\phi_s}P_{K^{(s)}}V_s.
\end{equation}
Successive application of these identities gives
\begin{equation}
\begin{aligned}
&
\widetilde{V}_r\cdots\widetilde{V}_1P_K
\\
&\qquad =
e^{i(\alpha_G+\sum_{s=1}^{r}\phi_s)}
P_{K^{(r)}}G.
\end{aligned}
\end{equation}
Therefore, every exact decomposition into verified primitives also
satisfies the local compilation identity. In the present
implementation, this argument is applied to
\[
\mathrm{CH},
\mathrm{CS},
CR_z(\pi/4),
\mathrm{CXX},
\mathrm{CCX},
\mathrm{CCZ},
\text{and Fredkin}.
\]
The fixed decompositions were verified by direct unitary multiplication. In particular, under the $\sqrt{X}$ convention used by Qiskit, the sequence implemented by \Call{EvaluateCCX}{} satisfies
\[
V_{\mathrm{CCX}}
=
e^{i3\pi/8}\mathrm{CCX}.
\]
The remaining decompositions are exact up to a global phase. Consequently, their compiled implementations satisfy Eq.~\eqref{eq:appendix-local-identity}.

\section{Detailed client-side gate compilation rules}
\label{app:gate-compilation-rules}

This appendix provides the complete client-side compilation routines used in the QOTPH implementation. These routines operationalize the gate rules summarized in Table~\ref{homomorphic-decompositions}. All key-dependent angle corrections, decompositions, and QOTP key-frame updates are computed locally by the client before the compiled circuit is submitted to the quantum server.

\begin{breakablealgorithm}
\caption{Client-side homomorphic gate compilation rules}
\label{alg:qhe-gate-eval}

{\scriptsize
\begin{algorithmic}[1]

\Procedure{EvaluateGate}{$G_\ell,Q_\ell,\Theta_\ell,K$}
\Statex \textbf{Execution domain: client-side compilation.}
\Statex This routine is executed locally by the client while constructing the compiled circuit.
\Statex The quantum server does not execute this routine and does not receive the QOTP key frame $K$.
\Statex In this algorithm, ``Apply'' means that the client appends the corresponding physical gate or decomposition to the compiled circuit.
\State \Comment{All key-dependent gate choices are resolved locally before server-side execution}

\If{$G_\ell=R_z$}
\State Let $Q_\ell=(q)$ and $\Theta_\ell=(\theta)$
\State Let $K_q=(j_q,k_q)$
\State $\theta'\gets (-1)^{j_q}\theta$
\State Apply $R_z(\theta')$ to qubit $q$

\ElsIf{$G_\ell=R_y$}
\State Let $Q_\ell=(q)$ and $\Theta_\ell=(\theta)$
\State Let $K_q=(j_q,k_q)$
\State $\theta'\gets (-1)^{j_q\oplus k_q}\theta$
\State Apply $R_y(\theta')$ to qubit $q$

\ElsIf{$G_\ell=R_x$}
\State Let $Q_\ell=(q)$ and $\Theta_\ell=(\theta)$
\State Apply $H$ to qubit $q$
\State $K_q\gets\Call{UpdateH}{K_q}$
\State Let $K_q=(j_q,k_q)$
\State $\theta'\gets (-1)^{j_q}\theta$
\State Apply $R_z(\theta')$ to qubit $q$
\State Apply $H$ to qubit $q$
\State $K_q\gets\Call{UpdateH}{K_q}$

\ElsIf{$G_\ell=H$}
\State Let $Q_\ell=(q)$
\State Apply $H$ to qubit $q$
\State $K_q\gets\Call{UpdateH}{K_q}$

\ElsIf{$G_\ell=S$}
\State Let $Q_\ell=(q)$
\State Apply $S$ to qubit $q$
\State $K_q\gets\Call{UpdateS}{K_q}$

\ElsIf{$G_\ell=S^\dagger$}
\State Let $Q_\ell=(q)$
\State $K\gets\Call{EvaluateGate}{R_z,(q),(-\pi/2),K}$

\ElsIf{$G_\ell=\sqrt{X}$}
\State Let $Q_\ell=(q)$
\State Apply $\sqrt{X}$ to qubit $q$
\State $K_q\gets\Call{UpdateSX}{K_q}$

\ElsIf{$G_\ell=T$}
\State Let $Q_\ell=(q)$
\State $K\gets\Call{EvaluateGate}{R_z,(q),(\pi/4),K}$

\ElsIf{$G_\ell=T^\dagger$}
\State Let $Q_\ell=(q)$
\State $K\gets\Call{EvaluateGate}{R_z,(q),(-\pi/4),K}$

\ElsIf{$G_\ell=X$}
\State Let $Q_\ell=(q)$
\State Apply $X$ to qubit $q$
\State \Comment{No QOTP key update is required}

\ElsIf{$G_\ell=Y$}
\State Let $Q_\ell=(q)$
\State Apply $Y$ to qubit $q$
\State \Comment{No QOTP key update is required}

\ElsIf{$G_\ell=Z$}
\State Let $Q_\ell=(q)$
\State Apply $Z$ to qubit $q$
\State \Comment{No QOTP key update is required}

\ElsIf{$G_\ell=CX$}
\State Let $Q_\ell=(c,t)$
\State Apply $CX(c,t)$
\State $(K_c,K_t)\gets\Call{UpdateCX}{K_c,K_t}$

\ElsIf{$G_\ell=CZ$}
\State Let $Q_\ell=(c,t)$
\State Apply $CZ(c,t)$
\State $(K_c,K_t)\gets\Call{UpdateCZ}{K_c,K_t}$

\ElsIf{$G_\ell=CY$}
\State Let $Q_\ell=(c,t)$
\State Apply $CY(c,t)$
\State $K_t\gets\Call{UpdateS}{K_t}$
\State $(K_c,K_t)\gets\Call{UpdateCX}{K_c,K_t}$
\State $K_t\gets\Call{UpdateS}{K_t}$

\ElsIf{$G_\ell=CS$}
\State Let $Q_\ell=(c,t)$
\State $K\gets\Call{EvaluateGate}{T,(c),(),K}$
\State $K\gets\Call{EvaluateGate}{T,(t),(),K}$
\State Apply $CX(c,t)$
\State $(K_c,K_t)\gets\Call{UpdateCX}{K_c,K_t}$
\State $K\gets\Call{EvaluateGate}{T^\dagger,(t),(),K}$
\State Apply $CX(c,t)$
\State $(K_c,K_t)\gets\Call{UpdateCX}{K_c,K_t}$
\State \Comment{$CS$ is compiled through an exact decomposition into verified primitives}

\ElsIf{$G_\ell=SWAP$}
\State Let $Q_\ell=(q_1,q_2)$
\State Apply $SWAP(q_1,q_2)$
\State $(K_{q_1},K_{q_2})\gets(K_{q_2},K_{q_1})$

\ElsIf{$G_\ell=CXX$}
\State Let $Q_\ell=(c,t_1,t_2)$
\State Apply $CX(c,t_1)$
\State $(K_c,K_{t_1})\gets\Call{UpdateCX}{K_c,K_{t_1}}$
\State Apply $CX(c,t_2)$
\State $(K_c,K_{t_2})\gets\Call{UpdateCX}{K_c,K_{t_2}}$

\ElsIf{$G_\ell=CH$}
\State Let $Q_\ell=(c,t)$
\State $K\gets\Call{EvaluateGate}{R_y,(t),(\pi/4),K}$
\State Apply $CX(c,t)$
\State $(K_c,K_t)\gets\Call{UpdateCX}{K_c,K_t}$
\State $K\gets\Call{EvaluateGate}{R_y,(t),(-\pi/4),K}$

\ElsIf{$G_\ell=CR_z$}
\State Let $Q_\ell=(c,t)$ and $\Theta_\ell=(\phi)$

\If{$\phi=\pi/4$}
    \State $K\gets\Call{EvaluateCRzQuarter}{c,t,K}$
    \State \Comment{Specialized exact decomposition for $CR_z(\pi/4)$}
\Else
    \State Let $K_c=(j_c,k_c)$ and $K_t=(j_t,k_t)$
    \State $\phi_t\gets(-1)^{j_t}\phi$
    \If{$j_c=0$}
        \State Apply $CR_z(\phi_t)$ with control $c$ and target $t$
    \Else
        \State Apply $R_z(\phi_t)$ to qubit $t$
        \State Apply $CR_z(-\phi_t)$ with control $c$ and target $t$
    \EndIf
    \State \Comment{No QOTP key update is required}
\EndIf

\ElsIf{$G_\ell=CR_x$}
\State Let $Q_\ell=(c,t)$ and $\Theta_\ell=(\theta)$
\State Let $K_c=(j_c,k_c)$ and $K_t=(j_t,k_t)$
\State $\theta_t\gets(-1)^{k_t}\theta$
\If{$j_c=0$}
\State Apply $CR_x(\theta_t)$ with control $c$ and target $t$
\Else
\State Apply $R_x(\theta_t)$ to qubit $t$
\State Apply $CR_x(-\theta_t)$ with control $c$ and target $t$
\EndIf
\State \Comment{No QOTP key update is required}

\ElsIf{$G_\ell=CR_y$}
\State Let $Q_\ell=(c,t)$ and $\Theta_\ell=(\theta)$
\State Let $K_c=(j_c,k_c)$ and $K_t=(j_t,k_t)$
\State $\theta_t\gets(-1)^{j_t\oplus k_t}\theta$
\If{$j_c=0$}
\State Apply $CR_y(\theta_t)$ with control $c$ and target $t$
\Else
\State Apply $R_y(\theta_t)$ to qubit $t$
\State Apply $CR_y(-\theta_t)$ with control $c$ and target $t$
\EndIf
\State \Comment{No QOTP key update is required}

\ElsIf{$G_\ell=R_{XX}$}
\State Let $Q_\ell=(p,q)$ and $\Theta_\ell=(\theta)$
\State Let $K_p=(j_p,k_p)$ and $K_q=(j_q,k_q)$
\State $\theta'\gets(-1)^{k_p\oplus k_q}\theta$
\State Apply $R_{XX}(\theta')$ to qubits $p,q$
\State \Comment{No QOTP key update is required}

\ElsIf{$G_\ell=R_{YY}$}
\State Let $Q_\ell=(p,q)$ and $\Theta_\ell=(\theta)$
\State Let $K_p=(j_p,k_p)$ and $K_q=(j_q,k_q)$
\State $\theta'\gets(-1)^{j_p\oplus k_p\oplus j_q\oplus k_q}\theta$
\State Apply $R_{YY}(\theta')$ to qubits $p,q$
\State \Comment{No QOTP key update is required}

\ElsIf{$G_\ell=R_{ZZ}$}
\State Let $Q_\ell=(p,q)$ and $\Theta_\ell=(\theta)$
\State Let $K_p=(j_p,k_p)$ and $K_q=(j_q,k_q)$
\State $\theta'\gets(-1)^{j_p\oplus j_q}\theta$
\State Apply $R_{ZZ}(\theta')$ to qubits $p,q$
\State \Comment{No QOTP key update is required}

\ElsIf{$G_\ell=CCX$}
\State Let $Q_\ell=(a,b,t)$
\State $K\gets\Call{EvaluateCCX}{a,b,t,K}$

\ElsIf{$G_\ell=CCZ$}
\State Let $Q_\ell=(a,b,c)$
\State Apply $H$ to qubit $c$
\State $K_c\gets\Call{UpdateH}{K_c}$
\State $K\gets\Call{EvaluateCCX}{a,b,c,K}$
\State Apply $H$ to qubit $c$
\State $K_c\gets\Call{UpdateH}{K_c}$

\ElsIf{$G_\ell=\mathrm{Fredkin}$}
\State Let $Q_\ell=(c,a,b)$
\State Apply $CX(b,a)$
\State $(K_b,K_a)\gets\Call{UpdateCX}{K_b,K_a}$
\State $K\gets\Call{EvaluateCCX}{c,a,b,K}$
\State Apply $CX(b,a)$
\State $(K_b,K_a)\gets\Call{UpdateCX}{K_b,K_a}$

\Else
\State Raise error: unsupported operation $G_\ell$
\EndIf

\State \Return $K$
\EndProcedure

\end{algorithmic}
}
\end{breakablealgorithm}

\begin{breakablealgorithm}
\caption{Homomorphic Evaluation of Composite Gates}
\label{alg:qhe-composite-gates}
\begin{algorithmic}[1]
\scriptsize
\Procedure{EvaluateCRzQuarter}{$c,t,K$}
\State \Comment{$CR_z(\pi/4)$ is implemented using two CNOT gates and corrected target rotations}
\State $K\gets\Call{EvaluateGate}{R_z,(t),(\pi/8),K}$
\State Apply $CX(c,t)$
\State $(K_c,K_t)\gets\Call{UpdateCX}{K_c,K_t}$
\State $K\gets\Call{EvaluateGate}{R_z,(t),(-\pi/8),K}$
\State Apply $CX(c,t)$
\State $(K_c,K_t)\gets\Call{UpdateCX}{K_c,K_t}$
\State \Return $K$
\EndProcedure

\Procedure{EvaluateCCX}{$a,b,t,K$}
\State \Comment{First Hadamard on the target, using $H\simeq R_z(\pi/2)\sqrt{X}R_z(\pi/2)$}
\State $K\gets\Call{EvaluateGate}{R_z,(t),(\pi/2),K}$
\State Apply $\sqrt{X}$ to qubit $t$
\State $K_t\gets\Call{UpdateSX}{K_t}$
\State $K\gets\Call{EvaluateGate}{R_z,(t),(\pi/2),K}$

\State Apply $CX(b,t)$
\State $(K_b,K_t)\gets\Call{UpdateCX}{K_b,K_t}$
\State $K\gets\Call{EvaluateGate}{R_z,(t),(7\pi/4),K}$

\State Apply $CX(a,t)$
\State $(K_a,K_t)\gets\Call{UpdateCX}{K_a,K_t}$
\State $K\gets\Call{EvaluateGate}{R_z,(t),(\pi/4),K}$

\State Apply $CX(b,t)$
\State $(K_b,K_t)\gets\Call{UpdateCX}{K_b,K_t}$
\State $K\gets\Call{EvaluateGate}{R_z,(t),(7\pi/4),K}$

\State Apply $CX(a,t)$
\State $(K_a,K_t)\gets\Call{UpdateCX}{K_a,K_t}$
\State $K\gets\Call{EvaluateGate}{R_z,(b),(\pi/4),K}$
\State $K\gets\Call{EvaluateGate}{R_z,(t),(\pi/4),K}$

\State \Comment{Second Hadamard on the target}
\State $K\gets\Call{EvaluateGate}{R_z,(t),(\pi/2),K}$
\State Apply $\sqrt{X}$ to qubit $t$
\State $K_t\gets\Call{UpdateSX}{K_t}$
\State $K\gets\Call{EvaluateGate}{R_z,(t),(\pi/2),K}$

\State \Comment{Additional phase correction on the controls}
\State Apply $CX(a,b)$
\State $(K_a,K_b)\gets\Call{UpdateCX}{K_a,K_b}$
\State $K\gets\Call{EvaluateGate}{R_z,(a),(\pi/4),K}$
\State $K\gets\Call{EvaluateGate}{R_z,(b),(7\pi/4),K}$
\State Apply $CX(a,b)$
\State $(K_a,K_b)\gets\Call{UpdateCX}{K_a,K_b}$

\State \Return $K$
\EndProcedure

\end{algorithmic}
\end{breakablealgorithm}

\begin{breakablealgorithm}
\caption{Primitive QOTP Key-Frame Update Rules}
\label{alg:qotp-key-updates}
\scriptsize
\begin{algorithmic}[1]

\Procedure{UpdateH}{$K_q$}
\State Let $K_q=(j_q,k_q)$
\State \Return $(k_q,j_q)$
\EndProcedure

\Procedure{UpdateS}{$K_q$}
\State Let $K_q=(j_q,k_q)$
\State \Return $(j_q,j_q\oplus k_q)$
\EndProcedure

\Procedure{UpdateSX}{$K_q$}
\State Let $K_q=(j_q,k_q)$
\State \Return $(j_q\oplus k_q,k_q)$
\EndProcedure

\Procedure{UpdateCX}{$K_c,K_t$}
\State Let $K_c=(j_c,k_c)$ and $K_t=(j_t,k_t)$
\State $K_c'\gets(j_c,k_c\oplus k_t)$
\State $K_t'\gets(j_c\oplus j_t,k_t)$
\State \Return $(K_c',K_t')$
\EndProcedure

\Procedure{UpdateCZ}{$K_c,K_t$}
\State Let $K_c=(j_c,k_c)$ and $K_t=(j_t,k_t)$
\State $K_c'\gets(j_c,k_c\oplus j_t)$
\State $K_t'\gets(j_t,k_t\oplus j_c)$
\State \Return $(K_c',K_t')$
\EndProcedure

\end{algorithmic}
\end{breakablealgorithm}

\bibliographystyle{quantum}

\bibliography{sample}

\end{document}